\documentclass{cpbtex}

\usepackage[scheme=plain,UTF8] {ctex} 

\usepackage{graphicx}%
\usepackage{multirow}%
\usepackage{amsmath,amssymb,amsfonts}%
\usepackage{amsthm}%
\usepackage{mathrsfs}%
\usepackage[title]{appendix}%
\usepackage{color,xcolor}%
\usepackage{textcomp}%
\usepackage{manyfoot}%
\usepackage{algorithm}%
\usepackage{algorithmicx}%
\usepackage{algpseudocode}%
\usepackage{listings}%
\usepackage{adjustbox}  

\usepackage{tikz}
\usepackage{physics}
\usepackage{mathtools}
\usepackage{bm}
\usepackage{bbm}
\usepackage{diagbox}
\usepackage{soul}
\usepackage{array}
\usepackage{subeqnarray}
\usepackage{cases}
\usepackage{indentfirst} 
\usepackage{tabularx}   
\usepackage{booktabs}

\usepackage{subeqnarray}
\usepackage{cases}
\usepackage{indentfirst}
\usepackage{ragged2e} 
\usepackage{makecell}

\newcommand{\ER}{Erd\"{o}s-R\'{e}nyi }
\newcommand{\vgamma}{\bm{\gamma}}
\newcommand{\vbeta}{\bm{\beta}}

\newtheorem{definition}{Definition}
\newtheorem{theorem}{Theorem}

\begin{document}

\title{Auxiliary-qubit-free quantum approximate optimization algorithm for the minimum dominating set problem}

\author{Guanghui Li(李广辉)$^{1}$, Xiaohui Ni(倪晓慧)$^{1}$, Junjian Su(苏俊健)$^{1}$, Sujuan Qin(秦素娟)$^{1}$, \\ Fenzhuo Guo(郭奋卓)$^{1}$, Bingjie Xu(徐兵杰)$^{2}$, Wei Huang(黄伟)$^{2}$ \thanks{Corresponding author. E-mail:~huangwei096505@aliyun.com}, and Fei Gao(高飞)$^{1}$\thanks{Corresponding author. E-mail:~gaof@bupt.edu.cn}\\
$^{1}${State Key Laboratory of Networking and Switching Technology, Beijing University of Posts } 
\\ {and Telecommunications, Beijing, 100876, China}\\  
$^{2}${National Key Laboratory of Security Communication, Institute of Southwestern Communication, }
\\ {Chengdu, 610041, China} \\ 
} 

\maketitle

\begin{abstract}
Quantum Approximate Optimization Algorithm (QAOA) is a promising framework for solving combinatorial optimization problems on near-term quantum devices. One such problem is the Minimum Dominating Set (MDS), which is known to be NP-hard. Existing QAOA algorithms for this problem typically require numerous auxiliary qubits, which increases circuit overhead and hardware requirements. In this paper, we propose an auxiliary-qubit-free QAOA algorithm based on Hamiltonian evolution (AQFH-QAOA) for the MDS problem. Unlike previous studies that require numerous auxiliary qubits, our algorithm eliminates the need for auxiliary qubits, thus significantly reducing circuit overhead. In addition, we present an auxiliary-qubit-free optimized implementation of the previously proposed Guerrero's QAOA algorithm (AQFG-QAOA) by utilizing gate decomposition techniques. Through a detailed analysis of gate complexity, we evaluate the applicability of these two algorithms. Numerical experiments demonstrate that our proposed algorithm achieves competitive solution quality compared to existing QAOA algorithms, making it a promising candidate for implementation on near-term quantum devices.
\end{abstract}



\section{Introduction}
\label{sec1}

Quantum algorithms offer the potential for significant computational speedups over classical algorithms by exploiting fundamental quantum mechanical principles. These algorithms have demonstrated advantages in multiple fields, including cryptanalysis~\cite{wu2024enhanced}, mathematical optimization~\cite{zhang2024quantum,xu2024quafu}, and machine learning~\cite{lloyd2014quantum,wu2022quantum,zeguendry2023quantum,yu2025flexible,gao2023quantum,su2025topology}. However, many well-known quantum algorithms, such as Shor's algorithm~\cite{shor1994algorithms}, Grover's algorithm~\cite{grover1996fast}, and HHL algorithm~\cite{harrow2009quantum}, rely on fault-tolerant quantum hardware with a large number of high-fidelity qubits, which is beyond the capabilities of current quantum devices.

To circumvent these limitations, Variational Quantum Algorithms (VQAs) have emerged as a promising framework. VQAs work by combining Parameterized Quantum Circuits (PQCs) with classical optimization routines to iteratively minimize an objective function~\cite{cerezo2021variational}. This hybrid quantum-classical approach is well-suited for Noisy Intermediate-Scale Quantum (NISQ) devices~\cite{preskill2018quantum}. VQAs can be categorized by their application domains, including the Variational Quantum Eigensolver (VQE) for quantum chemistry~\cite{peruzzo2014variational} and the Quantum Approximate Optimization Algorithm (QAOA) for combinatorial optimization~\cite{farhi2014quantum}. 

QAOA was originally proposed for Unconstrained Combinatorial Optimization Problems (UCOPs) and has become a widely studied quantum optimization method. It encodes the objective function into an objective Hamiltonian whose ground state corresponds to the optimal solution and employs a PQC to find this state. The PQC is constructed by alternating the unitary evolution operator of the objective Hamiltonian and the mixing Hamiltonian. As a promising approach, QAOA has been successfully applied to various UCOPs, including tail assignment~\cite{vikstaal2020applying}, vehicle routing~\cite{azad2022solving}, and job shop scheduling~\cite{kurowski2023application}. For problems with hard constraints, the general strategy of QAOA is to utilize the penalty method to incorporate the constraints as penalty terms into the objective function to construct the objective Hamiltonian~\cite{hadfield2017quantum}. This strategy aims to penalize infeasible solutions to enforce the hard constraints. Although infeasible solutions that violate constraints may still appear, feasible solutions usually have a higher probability in the output~\cite{saleem2023approaches}. This extends the applicability of QAOA to a wider range of Constrained Combinatorial Optimization Problems (CCOPs)~\cite{pan2024solving,li2025performance,li2025quantum,ni2025progressive}.

The Minimum Dominating Set (MDS) problem is a representative CCOP with broad applications in fields such as wireless networks, biological networks, and social networks. Its goal is to find the smallest subset of vertices such that every vertex in the graph is either included in this subset or adjacent to at least one vertex in it~\cite{ore1962theory}. As an NP-hard problem, various exact algorithms and approximate algorithms have been proposed to solve it~\cite{fomin2005exact,casado2023iterated}. 

In addition to the classical algorithms, recent studies have also explored the potential of QAOA in solving the MDS problem. Dinneen~\textit{et al.} introduced surplus variables to convert the inequality constraints of the MDS problem into equalities, and then applied the penalty method to convert the problem into a UCOP, providing a Quadratic Unconstrained Binary Optimization (QUBO) model for the MDS problem~\cite{dinneen2017formulating}. Pan~\textit{et al.} adopted a similar approach but refined the penalty term range and reduced the number of surplus variables for vertices with small degrees~\cite{pan2024qubo}. Both methods can be directly mapped into the objective Hamiltonian to construct the corresponding quantum circuit. However, they require a large number of auxiliary qubits, with the worst-case number of qubits reaching $\mathcal{O} \left ( n+\log_{2}{n}  \right ) $ and $n+2m$, respectively, where $n$ and $m$ represent the number of vertices and edges. Unlike them, Guerrero proposed a QAOA that avoids introducing surplus variables by constructing two clauses to encode the objective function and hard constraints of the MDS problem~\cite{guerrero2020solving}. However, this algorithm requires multiple high-level multi-OR-controlled phase gates, whose worst-case decomposition requires $n$ auxiliary qubits, which poses a challenge to the limited qubit resources on current NISQ devices.

In this study, we propose an Auxiliary-Qubit-Free QAOA algorithm based on Hamiltonian evolution (AQFH-QAOA) for the MDS problem, which is specifically designed to address the qubit limitations of NISQ devices. Different from previous algorithms, the proposed algorithm transforms the inequality constraints of the MDS problem into equalities via Boolean algebra identities, thus eliminating the use of auxiliary qubits in the constructed quantum circuit. Furthermore, we provide an Auxiliary-Qubit-Free optimized implementation of Guerrero's QAOA (AQFG-QAOA) by exploiting gate decomposition techniques. Subsequently, we analyze and compare the gate complexity of both algorithms and discuss their respective applicability on representative 3-regular graphs and \ER (ER) random graphs. The conclusions show that for graphs with a smaller average degree, such as 3-regular graphs, the AQFH-QAOA algorithm is preferable due to its lower gate complexity. For \ER (ER) random graphs with an edge probability of 0.5, the preferable algorithm depends on the problem size. Numerical experiments demonstrate that the proposed algorithm achieves competitive solution quality compared to existing QAOA algorithms. Additionally, an ablation study based on multi-angle QAOA indicates that the performance of the AQFH-QAOA algorithm can be further improved by replacing shared circuit parameters with independent ones.

The rest of this paper is organized as follows: In Section~\ref{sec2}, we provide background information and a review of related work. In Section~\ref{sec3}, we present our proposed AQFH-QAOA algorithm for the MDS problem. In Section~\ref{sec4}, we introduce the AQFG-QAOA algorithm and perform an algorithm comparison. In Section~\ref{sec5}, we conduct numerical simulations and perform an ablation study to explore performance improvements. Finally, we conclude this work and provide prospects for future research in Section~\ref{sec6}.

\section{Background}
\label{sec2}

In this section, we provide some background information related to this study in three parts. First, we present the definition of the Minimum Dominating Set (MDS) problem. Second, we provide a brief introduction to the Quantum Approximate Optimization Algorithm (QAOA). Finally, we review existing research related to this work, focusing on three representative algorithms.

\subsection{Introduction to the MDS problem}
\label{sec2.1}

Given an undirected graph $G=\left (V, E \right) $, where $V$ is the set of vertices and $E$ is the set of edges, a subset $D\subseteq V$ is called a dominating set if every vertex $v \in V$ is either in $D$ or adjacent to at least one vertex in $D$. The goal of the MDS problem is to find a dominating set with the smallest number of vertices.

Let $x_{i} \in \left \{ 0,1 \right \} $ be a binary variable associated with each vertex $v_{i}\in V$, where $x_{i}=1$ if the vertex $v_{i}$ is included in the dominating set, and $x_{i}=0$ otherwise. $N(v_{i})$ represents the set of neighbors of vertex $v_{i}$. The MDS problem can be formulated as follows:
\begin{align}
    & \min_{x_{i} \in \left \{0,1\right\}}    \quad  \sum_{ v_{i} \in V} x_{i},  
    \label{eq: objective function} 
    \\
    & \ \quad \text{s.t.} \quad \quad x_{i} + \sum_{ v_{j} \in N(v_{i})} x_{j} \ge 1, \quad  \forall \, v_{i} \in V .
    \label{eq: mds constraint}
\end{align}

\subsection{Introduction to QAOA } 
\label{sec2.2}

The Quantum Approximate Optimization Algorithm (QAOA) is a Variational Quantum Algorithm (VQA) designed to tackle combinatorial optimization problems on near-term quantum devices. It was first proposed by Farhi~\textit{et al.}~\cite{farhi2014quantum} as a promising approach for finding approximate solutions to optimization problems. QAOA encodes an optimization problem as a quantum Hamiltonian whose ground state corresponds to the optimal solution of the problem. The core is to build a Parameterized Quantum Circuit (PQC) to prepare this ground state through a set of adjustable parameters. The algorithm starts from the ground state of an initial Hamiltonian and evolves the quantum state through an alternating sequence of two parameterized unitary operators. One of the operators is derived from the objective function of the problem and is referred to as the objective Hamiltonian. The other is called the mixing Hamiltonian, which promotes the exploration of the solution space. 

Formally, let $f(x)$ be an objective function, where $x \in \left \{ 0,1 \right \} ^{\otimes n} $ encodes a candidate solution. QAOA starts from an initial state $\left | \psi_{0} \right \rangle $, which is typically set to a uniform superposition of all computational basis states, i.e., $\left | + \right \rangle ^{\otimes n} $. For a given number of circuit layers $p$, the PQC prepares the quantum state:
\begin{equation}
    \left|\psi_{p}(\boldsymbol{\vec{\gamma }}, \boldsymbol{\vec{\beta }})\right\rangle=U_{M}\left(\beta_{p}\right) U_{P}\left(\gamma_{p}\right) \cdots U_{M}\left(\beta_{1}\right) U_{P}\left(\gamma_{1}\right) \left | \psi_{0} \right \rangle^{\otimes n},
    \label{eq: ansatz state}
\end{equation}
where $U_{P}\left(\gamma\right) = e^{-i\gamma H_{P}} $ is the phase separation operator, i.e., the unitary evolution operator of the objective Hamiltonian $H_{P}$, and $U_{M}\left(\gamma\right) = e^{-i\gamma H_{M}} $ is the mixing operator, i.e., the unitary evolution operator of the mixing Hamiltonian $H_{M}$. Typically, $H_{M}=\sum_{i=1}^{n} \sigma^{x}_{i}$, which is the sum of the Pauli-$X$ operators. The parameters $\boldsymbol{\vec{\gamma }} =\left ( \gamma _{1},\gamma _{2},...,\gamma _{p}  \right ) $ and $\boldsymbol{\vec{\beta }} =\left ( \beta_{1},\beta_{2},...,\beta_{p} \right )$ are optimized using a classical optimizer (e.g., Adam, Nelder-Mead, or COBYLA) to maximize (or minimize) the expected value 
\begin{equation}
    F_{p}(\boldsymbol{\vec{\gamma}}, \boldsymbol{\vec{\beta}}) = \left\langle\psi_{p}(\boldsymbol{\vec{\gamma}}, \boldsymbol{\vec{\beta}})\right| H_{P}\left|\psi_{p}(\boldsymbol{\vec{\gamma}}, \boldsymbol{\vec{\beta}})\right\rangle.
    \label{eq: expected value}
\end{equation}

This quantum–classical hybrid loop continues until the expected value converges or the parameter optimization reaches a predefined number of iterations. The expressive power of PQC increases with the number of layers $p$. In theory, sufficiently large $p$ allows the quantum circuit to approximate the optimal solution arbitrarily well. In practice, circuit depth is constrained by hardware limitations. Despite this, even QAOA with low depth has been shown to perform competitively with classical heuristics on certain problem instances~\cite{zhou2020quantum}.

\subsection{Literature Review}
\label{sec2.3}

This subsection reviews prior works on QAOA for the MDS problem. Dinneen and Hua~\cite{dinneen2017formulating}, as well as Pan and Lu~\cite{pan2024qubo}, encode the MDS problem as a QUBO model but do not explicitly construct the corresponding objective Hamiltonian, although the derivation is straightforward. Nonetheless, their work provides a modeling strategy that bridges classical combinatorial optimization problems and variational quantum algorithms. Different from these two methods, Guerrero designed a QAOA algorithm for the MDS problem without using the QUBO model~\cite{guerrero2020solving}. While this algorithm also requires auxiliary qubits, it saves more circuit overhead compared to the previous two methods. Collectively, these studies provide valuable modeling strategies and confirm the potential of QAOA in addressing the MDS problem.

\subsubsection{Review of Dinneen and Hua's QAOA for the MDS problem}  
\label{sec2.3.1}

To solve the MDS problem using standard QAOA, it must first be reformulated as an Unconstrained Combinatorial Optimization Problem (UCOP). A common approach is the penalty method, which incorporates the problem's constraints into the objective function to penalize violations of the constraints. Since the constraints in the MDS problem are inequalities, Dinneen and Hua introduced surplus variables to convert them into equalities, then derive a QUBO formulation for the MDS problem, making it compatible with the standard QAOA~\cite{dinneen2017formulating}.

The objective function to be minimized is:
\begin{equation}
    F(x)=\sum_{v_{i} \in V} x_{i}+P \sum_{v_{i} \in V} p_{i},
    \label{eq: QUBO formulation of Dinneen and Hua}
\end{equation}
where
\begin{equation}
    p_{i}=\left(1-\left(x_{i}+\sum_{v_{j} \in N\left(v_{i}\right)}x_{j}\right)+\sum_{k=0}^{\left\lfloor\log_2 d_{i}\right\rfloor} 2^{k} y_{i, k}\right)^{2}.
    \label{eq: penalty term of Dinneen and Hua}
\end{equation}

Here, $x_{i} \in \left \{ 0, 1 \right \} $ indicates whether the vertex $v_{i}$ belongs to the dominating set ($x_{i} = 1$) or not ($x_{i} = 0$), the symbol $P$ represents the penalty coefficient, $p_{i}$ represents the penalty term, $N(v_{i})$ denotes the set of neighbors of vertex $v_{i}$ and $d_{i}$ denotes the degree of vertex $v_{i}$. Moreover, $y_{i,k}$ represents the $k$-th binary surplus variable introduced for vertex $v_i$, where $0 \leq k \leq \lfloor \log_2 d_i \rfloor$, $k \in \mathbb{Z} $.

The objective function $F(x)$ is designed to balance two goals: minimizing the number of selected vertices and penalizing vertices that violate the constraints. The first term $\sum_{v_i \in V} x_i$ encourages choosing fewer vertices, and the second term $P \sum_{v_i \in V} p_i$, scaled by a positive penalty constant $P > 1$, penalizes vertices that violate the dominating set constraint. The penalty term $p_i$ equals zero if vertex $v_i$ is dominated, i.e., $x_i = 1$ or has at least one neighbor $x_j = 1$. If vertex $v_i$ is not dominated, the surplus variable expression $\sum_{k=0}^{\left\lfloor\log_2 d_{i}\right\rfloor} 2^{k} y_{i, k}$ adjusts the penalty to balance the gap. This expression can represent integers exceeding $d_i$, ensuring that even in the extreme case where all vertices in the closed neighborhood $\{v_i\} \cup N(v_i)$ are selected, the penalty term can be offset by appropriately adjusting the value of $y_{i,k}$.

Consequently, each vertex requires $\lfloor \log_2d_i \rfloor + 1$ surplus variables to balance the gap to make the penalty term $p_i$ approach zero. The total number of binary variables is:
\begin{equation}
    n +\sum_{v_i \in V} \left( \lfloor \log_2d_i \rfloor + 1 \right) \leq \mathcal{O}(n + n \log_2 n).
\end{equation}
This means that the total number of binary variables used is at most $\mathcal{O}(n + n \log_2 n)$.

The QUBO formulation can be mapped to the objective Hamiltonian of QAOA by substituting the binary variables with Pauli-$Z$ operators. This mapping allows the quantum state to encode candidate solutions and evolve toward the low-energy configuration corresponding to the optimal solution of the minimum dominating set. Using the same initial state and mixing operators as the standard QAOA, the QAOA for solving the MDS problem can be constructed.

\subsubsection{Review of Pan and Lu's QAOA for the MDS problem}
\label{sec2.3.2}

Inspired by Krauss~\textit{et al.}'s use of quantum annealing to solve the maximum flow problem in QUBO form~\cite{krauss2020solving}, Pan and Lu proposed an improved QUBO formulation for the MDS problem~\cite{pan2024qubo}. Compared to Dinneen and Hua's model, their approach adjusts the coefficient of the last surplus variable in the surplus variable expression in Eq.~\ref{eq: penalty term of Dinneen and Hua} to correct its upper bound and refine its range. Additionally, they also reduce surplus variable usage by applying the penalty term proposed by Glover~\textit{et al.}~\cite{glover2022quantum} for vertices with degree $d_i = 0$ and $d_i = 1$, avoiding extra qubits for low-degree vertices.

Similar to Dinneen and Hua's method, they convert the inequality constraints in Eq.~\ref{eq: mds constraint} into equalities via surplus variables and incorporate them into the objective function using the penalty method. The objective function to be minimized is:
\begin{equation}
    F(x) = \sum_{v_i \in V} x_i + P \sum_{v_i \in V} p_i,
    \label{eq: QUBO formulation of Pan and Lu}
\end{equation}
where the penalty term $p_i$ depends on the degree $d_i$ of vertex $v_i$. The specific form is as follows:

\begin{itemize}
    \item When $d_i = 0$, the original constraint reduces to $x_i \geq 1$, and the corresponding penalty term is given by
    \begin{equation}
        p_i = (1 - x_i)^2.
        \label{eq: constraints when d_i = 0}
    \end{equation}
    
    \item When $d_i = 1$, the original constraint becomes $x_i + x_j \geq 1$, where $v_j$ is the sole neighbor of vertex $v_i$. Following the method proposed in Ref.~\cite{glover2022quantum}, the penalty term can be expressed as
    \begin{equation}
        p_i = \left(1 - (x_i + x_j) + x_i x_j\right)^2.
        \label{eq: constraints when d_i = 1}
    \end{equation}
    
    \item When $d_i \geq 2$, surplus variables are introduced to convert the inequality constraints into equalities. Let the surplus variable expression be denoted by $S$. For vertices with $d_i \geq 2$, the inequality constraint is transformed into $x_i + \sum_{v_j \in N(v_i)} x_j - S = 1$. Accordingly, the penalty term becomes
    \begin{equation}
        p_i = \left(1 - \left(x_i + \sum_{v_j \in N(v_i)} x_j\right) + S \right)^2.
        \label{eq: penalty term of Pan and Lu}
    \end{equation}
    The expression $S = \sum_{k=0}^{\lfloor \log_2d_i \rfloor - 1} 2^k y_{i,k} + \left(d_i + 1 - 2^{\lfloor \log_2d_i \rfloor} \right) y_{i,\lfloor \log_2d_i \rfloor}$, where $y_{i,k} \in \{0,1\}$ are auxiliary variables used to represent integers up to $d_i$ to offset the gap with the original constraints.
\end{itemize}

To summarize, vertices with $d_i < 2$ do not need to introduce any surplus variables, while those with $d_i \geq 2$ require $\lfloor \log_2d_i \rfloor + 1$ surplus variables. Thus, the total number of binary variables is less than $n + \sum_{v_i \in V} \left( \lfloor \log_2d_i \rfloor + 1 \right)$. Since
\begin{equation}
    n +\sum_{v_i \in V} \left( \lfloor \log_2d_i \rfloor + 1 \right) \leq n +\sum_{v_i \in V} d_i = n +2m,
\end{equation}
where $n$ and $m$ are the number of vertices and edges, respectively, the total number of binary variables in the worst case of this QUBO formulation is $n + 2m$. The QUBO formulation is mapped to an objective Hamiltonian and combined with the standard mixing Hamiltonian to construct a QAOA circuit, which requires at most $n + 2m$ qubits.

\subsubsection{Review of Guerrero's QAOA for the MDS problem}
\label{sec2.3.3}

Given a bitstring $x = x_0x_1x_2\ldots x_{n-1}$ representing $n$ vertices of a graph, where each bit $x_i \in \{0,1\}$ indicates whether the vertex $v_i$ is included in the dominating set ($x_i = 1$) or not ($x_i = 0$). Guerrero proposed the following objective function to be maximized~\cite{guerrero2020solving}:
\begin{equation}
    C(x)=\sum_{k=0}^{n-1} T_{k}(x) +D_{k}(x)
    \label{eq: objection function of Guerrero}
\end{equation}
where
\begin{equation}
T_{k}(x)=
\begin{cases} 
    1 ,\text{ if the $k$-th vertex is connected to some $i$-th vertex where $x_{i}=1$ }   \\
    0 ,\text{ if otherwise}
\end{cases}
\end{equation}
and
\begin{equation}
D_{k}(x)=
\begin{cases} 
    1  ,\text{ if  $x_{k} = 0$}   \\
    0  ,\text{ if  $x_{k} = 1$}
\end{cases}
\end{equation}

Here, the $T_k(x)$ clause indicates whether a vertex $v_k$ is dominated, so the term $\sum_{k=0}^{n-1} T_k(x)$ counts the total number of vertices that are dominated by the selected set. The $D_k(x)$ clause marks whether a vertex $v_k$ is excluded from the dominating set, so the term $\sum_{k=0}^{n-1} D_k(x)$ measures the size of the non-dominating set. Maximizing the term $\sum_{k=0}^{n-1} T_k(x)$ ensures that as many vertices as possible are dominated, satisfying the constraint of the MDS problem. Simultaneously, maximizing the term $\sum_{k=0}^{n-1} D_k(x)$ is equivalent to minimizing the size of the dominating set. Taken together, the objective function of the problem is formulated as maximizing the sum of these two terms, which balances constraint satisfaction with minimizing the size of the dominating set, aligning with the goal of the MDS problem.

Guerrero built a quantum circuit to implement this objective function, introducing an auxiliary qubit that remains unchanged throughout the circuit. The $D_k(x)$ clause is implemented using a single-controlled phase gate, with the qubit corresponding to vertex $v_k$ as the control qubit and the auxiliary qubit as the target qubit. The implementation of the $T_k(x)$ clause is more complicated, as it depends on the neighbors of vertex $v_k$. Specifically, each $T_k(x)$ clause is implemented through a multi-OR-controlled phase gate, where the qubits corresponding to the vertex $v_k$ and its neighbors are control qubits and the auxiliary qubit is the target qubit. For vertices with different neighbor structures, the control qubits are different.

The mixing Hamiltonian is constructed as in standard QAOA for the MaxCut problem, i.e., $H_M = \sum_{i=0}^{n-1} \sigma_i^x$, where $\sigma_i^x$ denotes the Pauli-$X$ operator acting on qubit $i$. The corresponding unitary evolution is given by $U_M(\beta) = e^{-i\beta H_M} = \prod_{i=0}^{n-1} R_{X}^{(i)}(2\beta)$, where $R_{X}^{(i)}(2\beta)$ represents a rotation operator around the x-axis. The initial state $\ket{\psi_0}$ is prepared as an equal superposition of all computational basis states over the $n$ qubits, which can be achieved by applying Hadamard gates, i.e., $\ket{\psi_0} = H^{\otimes n} \ket{0}^{\otimes n} = \ket{+}^{\otimes n}$. Both the $R_X^{(i)}(2\beta)$ gates and the Hadamard gates are only applied to the qubits corresponding to vertices, excluding the auxiliary qubits. Figure~\ref{Figure1} illustrates the quantum circuit constructed according to this algorithm for the 4-vertex 3-regular graph. All vertices in this graph share the same set of neighbors, so the quantum gate that implements the $T_k(x)$ clause is identical for each vertex.

\begin{figure}
    \centering
    \includegraphics[width=\linewidth]{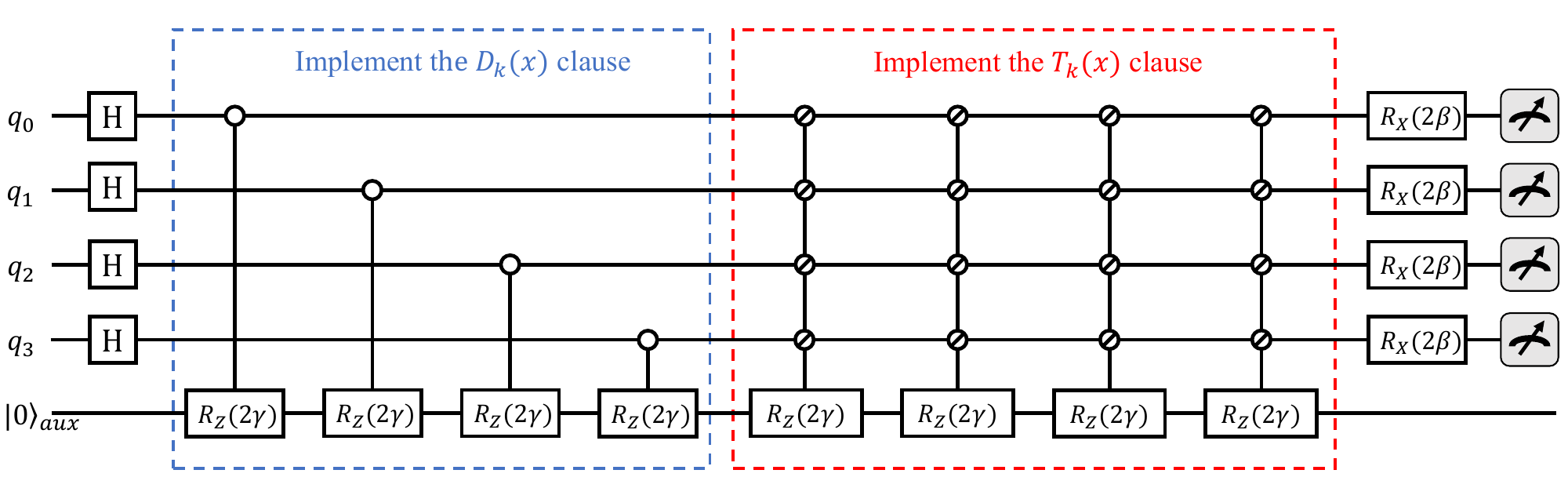}
    \caption{A single-layer quantum circuit for solving the MDS problem on a 4-vertex 3-regular graph, constructed according to Guerrero's QAOA algorithm~\cite{guerrero2020solving}. The qubit $q_k$ corresponds to vertex $v_k$ of the graph, and $\left | 0  \right \rangle_{aux} $ represents the auxiliary qubit. The blue dashed box represents the circuit part that implements all the $D_k(x)$ clauses, while the red dashed box indicates the circuit part that implements all the $T_k(x)$ clauses.} 
    \label{Figure1}
\end{figure}

Appendix~\ref{secA} presents Guerrero's decomposition method for the multi-OR-controlled phase gate and single-controlled phase gate. According to this decomposition, implementing a multi-OR controlled phase gate generally requires extra auxiliary qubits. In the worst case, apart from the one auxiliary qubit introduced previously for the $D_k(x)$ and $T_k(x)$ clauses, implementing the $T_k(x)$ clause corresponding to a vertex of degree $n-1$ requires $n-1$ additional auxiliary qubits. In total, the number of auxiliary qubits can reach $n$. Consequently, the overall qubit requirement of Guerrero's QAOA is at most $2n$.

\subsubsection{Summary of the number of qubits required for existing algorithms}
\label{sec2.3.4}

We summarize the number of qubits required by existing algorithms for the MDS problem in Table~\ref{tab1}, where $n$ and $m$ represent the number of vertices and edges, respectively. Without exception, the quantum circuits of these algorithms contain qubits corresponding to the vertices in the graph, and the number of auxiliary qubits used is related to the number of neighbors of each vertex. Therefore, we summarize the qubit number in the worst-case scenario for these algorithms based on the original literature.

\begin{table}[htbp]
    \centering
    \caption{The qubit number in the worst-case scenario of existing algorithms for solving the MDS problem}
    \vspace{8pt}
    \label{tab1}
    \begin{tabular}{c c c c}
    \hline
    Algorithm & Dinneen and Hua's QAOA & Pan and Lu's QAOA & Guerrero's QAOA  \\ 
    \hline
    Number of auxiliary qubits & $\mathcal{O}(n \log_2 n)$ & $2m$ & $n$  \\ 
    Total number of qubits & $\mathcal{O}(n + n \log_2 n)$ & $n+2m$ & $2n$ \\ 
    \hline
    \end{tabular}
\end{table}

\section{Auxiliary-qubit-free QAOA algorithm based on Hamiltonian evolution}
\label{sec3}

In this section, we introduce an Auxiliary-Qubit-Free QAOA algorithm based on Hamiltonian evolution (AQFH-QAOA) for solving the MDS problem. We first reformulate the integer programming model of the MDS problem, then construct the corresponding objective Hamiltonian. Using this objective Hamiltonian and the standard mixing Hamiltonian, we design a Parameterized Quantum Circuit (PQC) that implements the corresponding unitary evolutions. Finally, we summarize the overall algorithm procedure using pseudocode and analyze the gate complexity of the proposed algorithm.

\subsection{Transformation of integer programming model}
\label{sec3.1}

Consider an undirected graph $G = (V, E)$. Let $N(v_{i})$ denote the set of neighbors of vertex $v_{i} \in V$, $x_i \in \{0,1\}$ be the binary decision variable associated with vertex $v_i$, where $x_i = 1$ if vertex $v_i$ belongs to the dominating set $D$, and $x_i = 0$ otherwise. The integer programming model of the MDS problem is given in Section~\ref{sec2.1}. However, since the constraints are inequalities, this model is not compatible with the standard QAOA framework.

To address this issue, we use Boolean algebra formulations to transform the inequality constraints and reformulate the integer programming model. Inspired by the two clauses proposed by Guerrero~\cite{guerrero2020solving}, the model can be equivalently rewritten in the following form:
\begin{align}
    & \max_{x_{i} \in \left \{0,1\right\}}    \quad  \sum_{ v_{i} \in V} (1-x_{i}) ,  
    \label{eq: objective function1} 
    \\
    & \ \quad \text{s.t.} \quad \quad x_{i} \vee \left( \bigvee_{v_{j} \in N(v_{i})} x_{j} \right) = 1, \quad \forall \ v_{i} \in V .
    \label{eq: mds constraint1}
\end{align}

Eq.~\ref{eq: objective function1} reformulates the objective function expressed in Eq.~\ref{eq: objective function} by converting the problem of minimizing the size of the dominating set into the equivalent task of maximizing the size of the non-dominated set. Meanwhile, Eq.~\ref{eq: mds constraint1} rewrites the inequality constraints in Eq.~\ref{eq: mds constraint} as equality constraints using Boolean algebra. 

The Boolean expression $x_{i} \vee \left( \bigvee_{v_{j} \in N(v_{i})} x_{j} \right)$ evaluates to 1 if vertex $v_i$ is dominated by itself or at least one of its neighbors, and 0 otherwise. Since the maximum value of this expression is 1, the constraint in Eq.~\ref{eq: mds constraint1} can be equivalently enforced by maximizing its value. Accordingly, the integer programming model of the MDS problem is further expressed as follows:
\begin{align}
    & \max_{x_{i} \in \left \{0,1\right\}}  \quad  \sum_{ v_{i} \in V} (1-x_{i}) , 
    \label{eq: objective function2} 
    \\
    & \max_{x_{i} \in \left \{0,1\right\}}  \quad  \sum_{v_{i} \in V} \left( x_{i} \vee \left( \bigvee_{v_{j} \in N(v_{i})} x_{j} \right) \right).
    \label{eq: mds constraint2}
\end{align}

To convert the Boolean expression in Eq.~\ref{eq: mds constraint2} into an arithmetic expression form suitable for Hamiltonian construction, we prove the following theorem.

\begin{theorem}
    Let $x_i \in \{0, 1\}$ for all $v_i \in V$, the following identity holds: 
    \begin{equation}
        x_{i} \vee \left( \bigvee_{v_{j} \in N(v_{i})} x_{j} \right) = 1 - (1 - x_{i}) \prod_{v_{j} \in N(v_{i})} (1 - x_{j}).
    \end{equation}
    \label{theorem1}
\end{theorem}

\begin{proof}
We begin with a basic Boolean algebra identity: for any $a, b \in \{0, 1\}$, the OR operation satisfies
\begin{equation}
    a \lor b = 1 - (1 - a)(1 - b).
    \label{eq: eq1 in the proof}
\end{equation}

This identity can be verified by checking all possible values of $a$ and $b$. By induction, it generalizes to any finite disjunction:
\begin{equation}
    \bigvee_{k=1}^m x_k = 1 - \prod_{k=1}^m (1 - x_k),
    \label{eq: eq2 in the proof}
\end{equation}
where each $x_k \in \{0, 1\}$.

Now, let us define $S_{i} = \left\{ x_{j} \mid v_{j} \in N(v_{i}) \right\}$ to denote the set of variables corresponding to the neighbors of vertex $v_i$. Then the Boolean expression $x_i \lor \left( \bigvee_{x_j \in S_i} x_j \right)$ is equivalent to the disjunction over the set $\{ x_{i} \} \cup S_{i}$. By applying the generalized OR identity, we obtain
\begin{equation}
    x_i \lor \left( \bigvee_{x_j \in S_i} x_j \right) = 1 - \prod_{x \in \{x_i\} \cup S_i  } (1 - x).
    \label{eq: eq3 in the proof}
\end{equation}

The right-hand side of Eq.~\ref{eq: eq3 in the proof} can be rewritten as $1 - (1 - x_i) \prod_{x_j \in S_i} (1 - x_j)$, which completes the proof.

\end{proof}

Using the identity in Theorem~\ref{theorem1}, the objective function of the MDS problem can be reformulated as:
\begin{equation}
    \max_{x_{i} \in \left \{0,1\right\}} \quad \sum_{v_i \in V} (1 - x_i) + \lambda \sum_{v_i \in V} \left( 1 - (1 - x_i) \prod_{v_j \in N(v_i)} (1 - x_j) \right),
    \label{eq: maximize objective function}
\end{equation}
where the first term encodes the goal of maximizing the size of the non-dominated set and the second term introduces a penalty for non-dominated vertices. The penalty coefficient $\lambda > 1$ is typically chosen empirically to strike a balance between constraint enforcement and objective optimization.

\subsection{PQC construction }
\label{sec3.2}

In this section, we construct the PQC for the AQFH-QAOA algorithm, which consists of three fundamental components: the initial state, the phase separation operator, and the mixing operator.

\subsubsection{The initial state }
\label{sec3.2.1}

The initial state is set to the uniform superposition over all computational basis states:
\begin{equation}
    |\psi_0\rangle = |+\rangle^{\otimes n} = \left( \frac{|0\rangle + |1\rangle}{\sqrt{2}} \right)^{\otimes n}.
\end{equation}

This state is prepared by applying Hadamard gates to $\ket{0} ^{\otimes n}$, providing an unbiased starting point for the subsequent variational optimization.

\subsubsection{The phase separation operator }
\label{sec3.2.2}

The phase separation operator $U_P(\gamma)= e^{-i \gamma H_P}$ is defined as the unitary evolution of the objective Hamiltonian $H_P$. The objective Hamiltonian $H_P$ is derived from the transformed objective function of the MDS problem, which encodes both the objective and constraints of the MDS problem. 

Consider an undirected graph $G = (V, E)$, where $V = \{v_0, v_1, \ldots, v_{n-1}\}$, and suppose that the neighbors of vertex $v_i$ are $v_p, v_{p+1}, \ldots, v_q$. The objective function to be maximized in Eq.~\ref{eq: maximize objective function} can be equivalently reformulated as the objective function to be minimized as follows:
\begin{equation}
    \min_{x_{i} \in \left \{0,1\right\}} \quad -\sum_{i=0}^{n-1} \left(1 - x_i \right) - \lambda \sum_{i=0}^{n-1} \left( 1 - \left(1 - x_i\right) \prod_{j=p}^{q} \left(1 - x_j\right) \right).
    \label{eq: minimize objective function}
\end{equation}

Using the standard mapping $x_i = \frac{I - \sigma^{z}_{i}}{2}$~\cite{hadfield2021representation}, we replace the binary variable $x_i$ with $\frac{I - \sigma^{z}_{i}}{2}$, where $\sigma^{z}_{i}$ is the Pauli-$Z$ operator acting on qubit $i$. The objective Hamiltonian is then expressed as:
\begin{equation}
	\scalebox{1.0}{$\displaystyle
		\begin{aligned}
			H_P &= -\sum_{i=0}^{n-1} \left( I - \frac{I - \sigma_i^z}{2} \right) - \lambda \sum_{i=0}^{n-1} \left( I - \left(I - \frac{I - \sigma_i^z}{2}\right) \prod_{j=p}^{q} \left(I - \frac{I - \sigma_j^z}{2} \right) \right) \\[1.5ex]
			&= -\sum_{i=0}^{n-1} \frac{I + \sigma_i^z}{2} - \lambda \sum_{i=0}^{n-1} \left( I - \frac{I + \sigma_i^z}{2} \prod_{j=p}^{q} \frac{I + \sigma_j^z}{2} \right) \\[1.5ex]
			&= - \frac{n}{2} \, I - \frac{1}{2} \sum_{i=0}^{n-1} \sigma_i^z - \lambda n I + \frac{\lambda}{2} \sum_{i=0}^{n-1} \left( \left( I + \sigma_i^z\right)  \prod_{j=p}^{q} \frac{I + \sigma_j^z}{2} \right) \\[1.5ex]
			&= -\frac{(2\lambda + 1)n}{2} \, I - \frac{1}{2} \sum_{i=0}^{n-1} \sigma_i^z + \frac{\lambda}{2} \sum_{i=0}^{n-1} \left(  \frac{1}{2^{q - p + 1}} \left( I + \sigma_i^z\right)  \prod_{j=p}^{q} \left( I + \sigma_j^z \right)  \right).
			\label{eq: objective Hamiltonian in multiplication form}
		\end{aligned}
		$}
\end{equation}

By expanding the product terms and combining like terms, the objective Hamiltonian can be further simplified to:
\begin{equation}
	\begin{aligned}
		H_P & = -\frac{(2\lambda + 1)n}{2} \, I - \frac{1}{2} \sum_{i=0}^{n-1} \sigma_i^z + \frac{\lambda}{2} \sum_{i=0}^{n-1} \Biggl( \frac{1}{2^{q - p + 1}} \Biggl( I + \sigma_i^z + \sum_{j=p}^{q} \sigma_j^z + \sum_{j=p}^{q} \sigma_i^z \sigma_j^z \\
		& \qquad  + \sum_{p \le j < k \le q} \sigma_j^z \sigma_k^z  + \sum_{p \le j < k \le q} \sigma_i^z \sigma_j^z \sigma_k^z + \cdots + \sigma_i^z \sigma_p^z \cdots \sigma_q^z \Biggr) \Biggr)
		\\[1.5ex]
		& = -\frac{(2\lambda + 1)n}{2} \, I  + \frac{\lambda}{2} \sum_{i=0}^{n-1}  \frac{1}{2^{q - p + 1}}\, I   + \frac{1}{2} \sum_{i=0}^{n-1} \left( \frac{\lambda-2^{q - p + 1}}{2^{q - p + 1}}\right)  \sigma_i^z + \frac{\lambda}{2} \sum_{i=0}^{n-1} \Biggl( \frac{1}{2^{q - p + 1}} \cdot \\
		& \qquad   \Biggl( \sum_{j=p}^{q} \sigma_j^z + \sum_{j=p}^{q} \sigma_i^z \sigma_j^z+ \sum_{p \le j < k \le q} \sigma_j^z \sigma_k^z  + \sum_{p \le j < k \le q} \sigma_i^z \sigma_j^z \sigma_k^z + \cdots + \sigma_i^z \sigma_p^z \cdots \sigma_q^z \Biggr) \Biggr)
	\end{aligned}
	\label{eq: objective Hamiltonian in simplified form}
\end{equation}

After neglecting the constant term in the Hamiltonian, the phase separation operator can be derived as:
\begin{equation}
\scalebox{1.12}{$\displaystyle
	\begin{aligned}
		U_P(\gamma) & = e^{-i \frac{\gamma}{2} \sum_{i=0}^{n-1} \left( \frac{\lambda-2^{q - p + 1}}{2^{q - p + 1}}\right) \sigma_i^z} \\
		&\ \ \ \cdot e^{-i \frac{\gamma }{2} \lambda \sum_{i=0}^{n-1} \left( \frac{1}{2^{q-p+1}} \left(  \sum_{j=p}^{q} \sigma_j^z + \sum_{j=p}^{q} \sigma_i^z \sigma_j^z+ \sum_{p \le j < k \le q} \sigma_j^z \sigma_k^z  + \sum_{p \le j < k \le q} \sigma_i^z \sigma_j^z \sigma_k^z + \cdots + \sigma_i^z \sigma_p^z \cdots \sigma_q^z \right)  \right)}
		\label{eq: original phase separation operator}
	\end{aligned}
$}
\end{equation}

It is known that when $[A, B] = AB-BA = 0$, the identity $e^{A+B} = e^A \cdot e^B$ holds, and thus $e^{-i(A+B)t} = e^{-iAt} \cdot e^{-iBt}$ holds. Since all Pauli-$Z$ operators commute, the operator can be further factorized into the product of exponential terms, as follows:
\begin{equation}
\scalebox{1.12}{$\displaystyle
	\begin{aligned}
		U_P(\gamma) &= \prod_{i=0}^{n-1} e^{-i \frac{\gamma}{2} \left( \frac{\lambda-2^{q - p + 1}}{2^{q - p + 1}}\right) \sigma_i^z} \cdot  \prod_{i=0}^{n-1}  \Biggl( \prod_{j=p}^{q} e^{-i \frac{\gamma }{2}  \frac{\lambda}{2^{q-p+1}} \sigma_j^z} \cdot \prod_{j=p}^{q} e^{-i \frac{\gamma }{2} \frac{\lambda}{2^{q-p+1}} \sigma_i^z \sigma_j^z} \\
		&  \quad \cdot \prod_{p \le j < k \le q} e^{-i \frac{\gamma }{2} \frac{\lambda}{2^{q-p+1}} \sigma_j^z \sigma_k^z} \cdot \prod_{p \le j < k \le q} e^{-i \frac{\gamma }{2} \frac{\lambda}{2^{q-p+1}} \sigma_i^z \sigma_j^z \sigma_k^z} \cdots e^{-i \frac{\gamma }{2} \frac{\lambda}{2^{q-p+1}} \sigma_i^z \sigma_p^z \cdots \sigma_q^z}   \Biggr) 
	\label{eq: factorized Phase Separation Operator}
	\end{aligned}
$}
\end{equation}

Eq.\ref{eq: factorized Phase Separation Operator} contains multi-qubit operators whose number grows with the neighbor size of each vertex. By expressing these operators as quantum gates, it becomes:
\begin{equation}
\scalebox{1.12}{$\displaystyle
	\begin{aligned}
		U_P(\gamma) &= \prod_{i=0}^{n-1} R_Z^{(i)}\left(  \frac{\lambda-2^{q - p + 1}}{2^{q - p + 1}} \gamma \right)  \cdot \prod_{i=0}^{n-1} \Biggl( \prod_{j=p}^{q} R_Z^{(j)}\left( \frac{\lambda }{2^{q-p+1}} \gamma \right)  \cdot \prod_{j=p}^{q} R_{ZZ}^{(i,j)}\left( \frac{\lambda }{2^{q-p+1}} \gamma \right)
		\\
		& \quad \cdot  \prod_{p \le j < k \le q} R_{ZZ}^{(j,k)}\left( \frac{\lambda }{2^{q-p+1}} \gamma \right) \cdot  \prod_{p \le j < k \le q} R_{ZZZ}^{(i,j,k)}\left( \frac{\lambda }{2^{q-p+1}}\gamma \right) \cdots R_{ZZ \cdots Z}^{(i,p,\dots,q)}\left( \frac{\lambda }{2^{q-p+1}} \gamma \right) \Biggr)
	\label{eq: quantum gate form of phase separation operator}
	\end{aligned}
$}
\end{equation}

Here, $R_Z^{(i)}(\theta) = e^{ -i\theta \sigma _{i }^{z}/2 } $ denotes a single-qubit $R_Z$ gate, $R_{ZZ}^{(j,k)}(\theta)= e^{ -i\theta \sigma _{j}^{z}\sigma _{k}^{z}/2 }$ denotes a two-qubit $R_{ZZ}$ gate, and higher-order operators such as \( R_{ZZZ}^{(j,k,l)}(\theta) \) represent a multi-qubit $R_{ZZZ}$ gate. Figure~\ref{fig2} illustrates the quantum gate decomposition of several operators, from which quantum gates for higher-order operators can be constructed. In fact, each term in $U_P(\gamma)$ can be implemented using only CNOT gates and a single-qubit $R_Z$ gate, which makes it simple to construct a quantum circuit for the phase separation operator $U_P(\gamma)$. However, graphs with high connectivity will produce a large number of multi-qubit operator terms, resulting in excessive circuit depth, which slows down parameter optimization and increases susceptibility to errors.

\begin{figure}
    \centering
    \includegraphics[width=\linewidth]{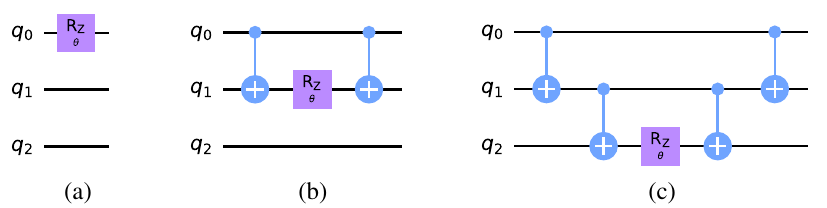}
    \caption{Quantum circuits for (a) the single-qubit operator $e^{-i \theta \sigma_0^z /2}$, (b) the two-qubit operator $e^{-i \theta \sigma_0^z \sigma_1^z /2 }$, and (c) the three-qubit operator $e^{-i \theta \sigma_0^z \sigma_1^z \sigma_{2}^z /2}$. Similarly, a $k$-qubit operator $e^{-i \theta \sigma_0^z \sigma_1^z \cdots \sigma_{k-1}^z}$ can be implemented using $2(k-1)$ CNOT gates and a single-qubit $R_Z$ gate~\cite{hadfield2021representation}.} 
    \label{fig2}
\end{figure}

\subsubsection{The mixing operator }
\label{sec3.2.3}

The mixing operator is defined as $U_M \left ( \beta \right ) = e^{-i \beta H_M}$, where $H_M$ is the mixing Hamiltonian that promotes transitions between quantum states. For the MDS problem, we adopt the standard transverse-field mixing Hamiltonian, i.e.,
\begin{equation}
    H_M = \sum_{i=0}^{n-1} \sigma^{x}_{i},
    \label{eq: mixing Hamiltonian}
\end{equation}
where $\sigma^{x}_{i}$ denotes the Pauli-$X$ operator acting on qubit $i$.

Then the mixing operator is given by:
\begin{equation}
    U_M(\beta)  = e^{-i \beta \sum_{i=0}^{n-1}  \sigma^{x}_{i}} = \prod_{i=0}^{n-1} e^{-i \beta \sigma^{x}_{i}} = \prod_{i=0}^{n-1} R_X^{(i)}(2\beta),
\end{equation}
where $\beta \in \left [ 0,\pi  \right ] $ is a variational parameter. Each term $R_X^{(i)}(2\beta)$ corresponds to a single-qubit $R_X$ gate acting on qubit $i$, enabling bit flips between $|0\rangle$ and $|1\rangle$. This mixing operator ensures that amplitude can be redistributed between different states, which is crucial for the algorithm to explore a wide region of Hilbert space.

\subsection{Algorithm description}
\label{sec3.3}

This subsection presents the pseudocode and procedural details of the proposed AQFH-QAOA algorithm. The complete pseudocode is provided in Algorithm~\ref{alg1}, which adheres to the standard QAOA framework. The algorithm proceeds as follows:

\begin{algorithm}
\caption{AQFH-QAOA algorithm for the MDS problem}
\label{alg1}
\begin{algorithmic}[1]
\renewcommand{\algorithmicrequire}{\textbf{Input:}}
\renewcommand{\algorithmicensure}{\textbf{Output:}}

\Require Graph $G = (V,E) $ with $\left | V \right | = n$; Number of circuit layer $p$; Maximum number of iterations $T$
\Ensure A bitstring $x^\ast \in \left \{ 0,1 \right \}^n$ representing the approximate solution to the MDS problem

\State Prepare the initial state $|\psi_0\rangle = |+\rangle^{\otimes n}$;
\State Construct the objective Hamiltonian $H_{P}$ according to the transformed integer programming model of the MDS problem; 
\State Set the mixing Hamiltonian to $H_M = \sum_{i=0}^{n-1} \sigma^x_i$;
\State Initialize the variational parameters $\boldsymbol{\vec{\gamma}} = (\gamma_1, ..., \gamma_p)$ and $\boldsymbol{\vec{\beta}} = (\beta_1, ..., \beta_p)$;

\While {the number of iterations does not reach $T$}  
    \For {$k = 1$ to $p$}  
        \State Add the quantum circuit of the phase separation operator $U_{P}(\gamma_{k})=e^{-i\gamma_{k}H_{P}}$;
        \State Add the quantum circuit of the mixing operator $U_M(\beta_{k}) = e^{-i\beta_{k} H_M}$;
    \EndFor

    \State Measure the final state $|\psi_p(\boldsymbol{\vec{\gamma}}, \boldsymbol{\vec{\beta}})\rangle$ in the computational basis and evaluate the expected value $F_{p}(\boldsymbol{\vec{\gamma}}, \boldsymbol{\vec{\beta}})$;
    \State Update the parameters $\boldsymbol{\vec{\gamma}}$ and $\boldsymbol{\vec{\beta}}$ using a classical optimizer to minimize $F_{p}(\boldsymbol{\vec{\gamma}}, \boldsymbol{\vec{\beta}})$;
\EndWhile

\State Return a bitstring $x^\ast = x^\ast_{0}x^\ast_{1}...x^\ast_{n-1}$.

\end{algorithmic}
\end{algorithm}

\textbf{Step 1}: Prepare the initial state as a uniform superposition over all computational basis states, i.e., $|\psi_0\rangle = |+\rangle^{\otimes n}$;

\textbf{Step 2}: Construct the objective Hamiltonian $H_P$ based on the transformed integer programming model of the MDS problem given in Section~\ref{sec3.1};

\textbf{Step 3}: Set the mixing Hamiltonian to the standard transverse-field Hamiltonian $H_M = \sum_{i=0}^{n-1} \sigma^x_i$;

\textbf{Step 4}: Initialize the variational parameters $\boldsymbol{\vec{\gamma}} = (\gamma_1, ..., \gamma_p)$ and $\boldsymbol{\vec{\beta}} = (\beta_1, ..., \beta_p)$, with $\gamma_{k} \in [0,2\pi]$ and $\beta_{k} \in [0,\pi]$ for $k = 1, \ldots, p$;

\textbf{Step 5}: For each layer $k = 1, \dots, p$, alternately apply the phase separation operator $U_P(\gamma_k) = e^{-i\gamma_k H_P}$ and the mixing operator $U_M(\beta_k) = e^{-i\beta_k H_M}$ to construct the quantum state $|\psi(\boldsymbol{\vec{\gamma}}, \boldsymbol{\vec{\beta}})\rangle = \big(\prod_{k=1}^{p} U_M(\beta_k)U_P(\gamma_k) \big)|\psi_0\rangle$;

\textbf{Step 6}: Measure the final state $|\psi_p(\boldsymbol{\vec{\gamma}}, \boldsymbol{\vec{\beta}})\rangle$ of the quantum circuit in the computational basis to obtain a bitstring $\tilde{x} \in \left \{ 0,1 \right \}^n $ and evaluate the expected value $F_{p}(\boldsymbol{\vec{\gamma}}, \boldsymbol{\vec{\beta}}) = \langle \psi_p(\boldsymbol{\vec{\gamma}}, \boldsymbol{\vec{\beta}}) | H_P | \psi_p(\boldsymbol{\vec{\gamma}}, \boldsymbol{\vec{\beta}}) \rangle$;

\textbf{Step 7}: Update the parameters $\boldsymbol{\vec{\gamma}}$ and $\boldsymbol{\vec{\beta}}$ using a classical optimizer (e.g., Adam, COBYLA, or Nelder-Mead) to minimize $F_{p}(\boldsymbol{\vec{\gamma}}, \boldsymbol{\vec{\beta}})$;

\textbf{Step 8}: Repeat Steps 5–7 until the maximum number of iterations $T$ is reached;

\textbf{Step 9}: After the optimization is completed, output the bitstring $x^\ast \in \left \{ 0,1 \right \}^n$ with the lowest cost, which approximates the solution to the MDS problem.

\subsection{Gate complexity}
\label{sec3.4}

Next, we analyze the gate complexity of the AQFH-QAOA algorithm. Given an undirected graph $G = (V, E)$ with $n$ vertices, the degree of each vertex is usually not equal. However, for simplicity of analysis, we assume that each vertex has $d$ neighbors, denoted by $v_p, v_{p+1}, \cdots, v_q$, such that $d = q-p+1$.

The number of quantum gates in the quantum circuit of the phase separation operator $U_P(\gamma)$ can be estimated based on Eq.~\ref{eq: quantum gate form of phase separation operator}. Some of these quantum gates can be merged. For example, a single-qubit $R_Z^{(j)}$ gate can always be eliminated by combining the angle parameter with the $R_Z^{(i)}$ gate acting on each qubit. However, this cannot be generalized for two- (multi-)qubit gates, since their acting qubits are not necessarily the same. Based on this, we can obtain the worst-case number of quantum gates in the quantum circuit of the phase separation operator $U_P(\gamma)$:
\begin{equation}
	\scalebox{1.0}{$\displaystyle
		\begin{aligned}
			N_{gate} \le  \ n\cdot R_Z + n\cdot \left ( C_{d}^{1}\cdot R_{ZZ} +C_{d}^{2}\cdot R_{ZZ} + C_{d}^{2}\cdot R_{ZZZ} +\cdots+ C_{d}^{d}\cdot R_{ZZ\cdots Z}  \right )
			\label{eq: gate count}
		\end{aligned}
	$}
\end{equation}	
For clarity, we have omitted the angle parameters of the quantum gates and the specific qubits they act on in the above equation.

As illustrated in Figure~\ref{fig2}, an $R_Z$ gate can be implemented with only a single-qubit gate, an $R_{ZZ}$ gate can be decomposed into 2 CNOT gates and a single-qubit $R_Z$ gate, and an $R_{ZZZ}$ gate can be decomposed into 4 CNOT gates and a single-qubit $R_Z$ gate. More generally, a $k$-qubit $R_{Z \cdots Z}$ gate can be decomposed into $2(k-1)$ CNOT gates and a single-qubit $R_Z$ gate. Based on these decompositions, we can obtain the worst-case number of elementary gates (CNOT gate and single-qubit gates) in the quantum circuit of the phase separation operator in Eq.~\ref{eq: quantum gate form of phase separation operator}. That is
\begin{equation}
	\scalebox{1.0}{$\displaystyle
		\begin{aligned}
			N_{CNOT} \le & \ 0\cdot n + n\cdot \left (2 \cdot C_{d}^{1}+ 2\cdot C_{d}^{2} + 4\cdot C_{d}^{2} + 4 \cdot C_{d}^{3}+\cdots+  2d \cdot C_{d}^{d}  \right )
			\\
			= & \ n \left( d-1 \right) \cdot 2^{d+1} + 2n
			\label{eq: CNOT gate count}
		\end{aligned}
		$}
\end{equation}

\begin{equation}
	\scalebox{1.0}{$\displaystyle
		\begin{aligned}
			N_{Single-qubit} \le & \ 1\cdot n + n\cdot \left (1 \cdot C_{d}^{1}+ 1\cdot C_{d}^{2} + 1\cdot C_{d}^{2} + 1\cdot C_{d}^{3}+ 1 \cdot C_{d}^{3}+\cdots+  1\cdot C_{d}^{d}  \right )
			\\
			= & \ n \cdot 2^{d+1}-n\left (d+1 \right ) 
			\label{eq: Single-qubit gate count}
		\end{aligned}
		$}
\end{equation}

Including the $n$ Hadamard gates for preparing the initial state and the $n$ single-qubit $R_X$ gates of the mixing operator, a single-layer quantum circuit constructed by the AQFH-QAOA algorithm requires at most
$n (d-1) \cdot 2^{d+1} +2n$ CNOT gates and $n \cdot 2^{d+1}-n\left (d-1 \right )$ single-qubit gates, where $n$ denotes the number of vertices and $d$ represents the average degree of the vertices. This indicates that the total number of elementary gates used by the algorithm grows polynomially with the number of vertices and exponentially with the average degree of the vertices. 

Appendix~\ref{secA} presents the number of elementary gates in a single-layer quantum circuit constructed using Guerrero's QAOA algorithm for an $n$-vertex undirected graph with an average vertex degree of $d$. Specifically, the single-layer quantum circuit contains $20nd-6n$ CNOT gates and $16nd$ single-qubit gates. As noted in Ref.~\cite{maslov2022depth}, two-qubit gates such as CNOT gates typically dominate circuit depth due to their slower execution times and higher error rates, while single-qubit gates are relatively fast and can often be executed in parallel, contributing negligibly to the overall depth. Accordingly, in this work, we use the number of CNOT gates as a proxy for circuit depth, neglecting the contribution of single-qubit gates.

By approximating $N_{CNOT}$ (in Eq.~\ref{eq: CNOT gate count}) to $20nd-6n$, we can estimate the average vertex degree $d$ in an undirected graph when the AQFH-QAOA algorithm loses its advantage over Guerrero's QAOA using auxiliary qubits. That is
\begin{equation}
	\begin{aligned}
	   n \left( d-1 \right) \cdot 2^{d+1} + 2n \ \approx \  20nd-6n 
	\label{eq: CNOT count estimate}
	\end{aligned}
\end{equation}

Eq.~\ref{eq: CNOT count estimate} can be solved approximately using the bisection method or linear interpolation, yielding a numerical solution $d \approx 3.62$. This indicates that when the average vertex degree does not exceed 3.62, the AQFH-QAOA algorithm has significant advantages over Guerrero's QAOA algorithm in terms of both auxiliary qubits and gate count. In fact, when $d$ is small, the total number of elementary gates of the AQFH-QAOA algorithm can be approximated as $ \mathcal{O}(n)$. For example, for 3-regular graphs, a single-layer quantum circuit constructed by this algorithm requires at most $34n$ CNOT gates and $14n$ single-qubit gates.

\section{Auxiliary-qubit-free optimized implementation algorithm of Guerrero's QAOA}
\label{sec4}

In this section, we apply gate decomposition techniques without auxiliary qubits to implement Guerrero's QAOA as described in Section~\ref{sec2.3.3}. To distinguish this algorithm from Guerrero's QAOA, we refer to it as the AQFG-QAOA algorithm. After decomposing the quantum circuit of this algorithm into elementary gates, we compare its gate complexity with that of the AQFH-QAOA algorithm.

\subsection{Decomposition of multi-OR-controlled phase gates without auxiliary qubits }
\label{sec4.1}

Apart from the one auxiliary qubit used to encode the $T_k(x)$ and $D_k(x)$ clauses, the additional auxiliary qubits in Guerrero's QAOA algorithm come from the decomposition of multi-OR-controlled phase gates, as discussed in Appendix~\ref{secA}. However, such a gate can first be decomposed into a sequence of single-qubit gates and Multi-Controlled $X$ (MCX) gates using the method illustrated in Figure~\ref{fig3}. Subsequently, the resulting MCX gates can be further decomposed into elementary gates of quadratic size without auxiliary qubits~\cite{barenco1995elementary}. Throughout this section, the term ``auxiliary qubit'' excludes the qubit used to encode the $T_k(x)$ and $D_k(x)$ clauses. Accordingly, we consider an $(n+1)$-qubit circuit, where the qubit $\left | 0  \right \rangle _{aux}$ associated with the $T_k(x)$ and $D_k(x)$ clauses is not counted as an auxiliary qubit in our analysis.

Barenco~\textit{et al.} pioneered a systematic method for decomposing MCX gates in Corollary 7.6 of Ref.~\cite{barenco1995elementary}. Subsequently, Refs.~\cite{liu2008analytic,saeedi2013linear,luo2016comment,da2022linear} also reported methods for decomposing MCX gates into elementary gates of quadratic size without auxiliary qubits. These studies aimed to optimize the circuit depth and circuit size of the decomposition proposed by Barenco~\textit{et al.}, but most of them lack precise results regarding the number of elementary gates after decomposition. One notable exception is the work of Liu~\textit{et al.}~\cite{liu2008analytic}, which presents the exact number of CNOT gates and single-qubit gates obtained from MCX gate decomposition. Subsequently, we will evaluate the number of elementary gates after multi-OR-controlled phase gate decomposition using the exact number of elementary gates proposed in the paper of Liu~\textit{et al.}.

\begin{figure}
    \centering
    \includegraphics[width=0.75\linewidth]{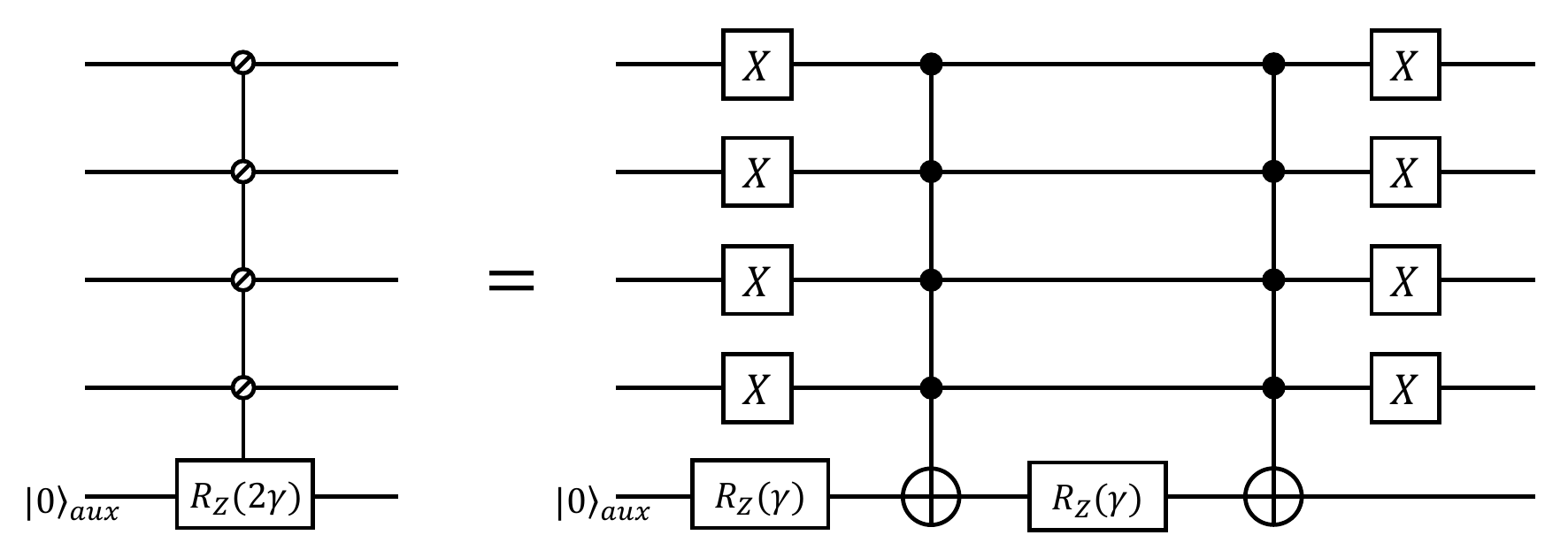}
    \caption{ Decomposition of a multi-OR-controlled phase gate with $n$ control qubits using $2n+2$ single-qubit gates and two $n$-controlled $X$ gates without auxiliary qubits. Here, $n=4$. It employs the relations $R_Z(\theta)R_Z(\theta)=R_Z(2\theta)$, $XR_Z(\theta)X=R_Z(-\theta)$ and $R_Z(\theta)R_Z(-\theta)=I$.}
    \label{fig3}
\end{figure}

For completeness, we briefly review the decomposition of the MCX gate as described in Refs.~\cite{barenco1995elementary,liu2008analytic}. The first step is to decompose an MCX gate with $n$ control qubits (denoted $C^n(X)$) into two $(n-1)$-controlled $X$ gates ($C^{n-1}(X)$), one $(n-1)$-controlled $\sqrt{X}$ gate ($C^{n-1}(\sqrt{X})$), and two single-controlled gates using the circuit identity shown in Figure~\ref{fig4}. In this decomposition, $\sqrt{X}$ denotes a square root of the Pauli-$X$ operator, satisfying $\sqrt{X}^2 = X$. More precisely, $\sqrt{X} = \left ( 1-i \right ) \left ( I+iX \right ) /2$. We note a difference in notation between this work and Ref.~\cite{barenco1995elementary,liu2008analytic}. In their paper, $C^n(X)$ refers to an $n$-qubit gate with $n-1$ control qubits and one target qubit, while in this paper, $C^n(X)$ refers to an $(n+1)$-qubit gate with $n$ control qubits. This distinction arises because the original quantum circuit in Guerrero's QAOA algorithm introduces an auxiliary qubit as the target qubit, while the qubits corresponding to a vertex and its neighbors act as the control qubits.

\begin{figure}
    \centering
    \includegraphics[width=0.7\linewidth]{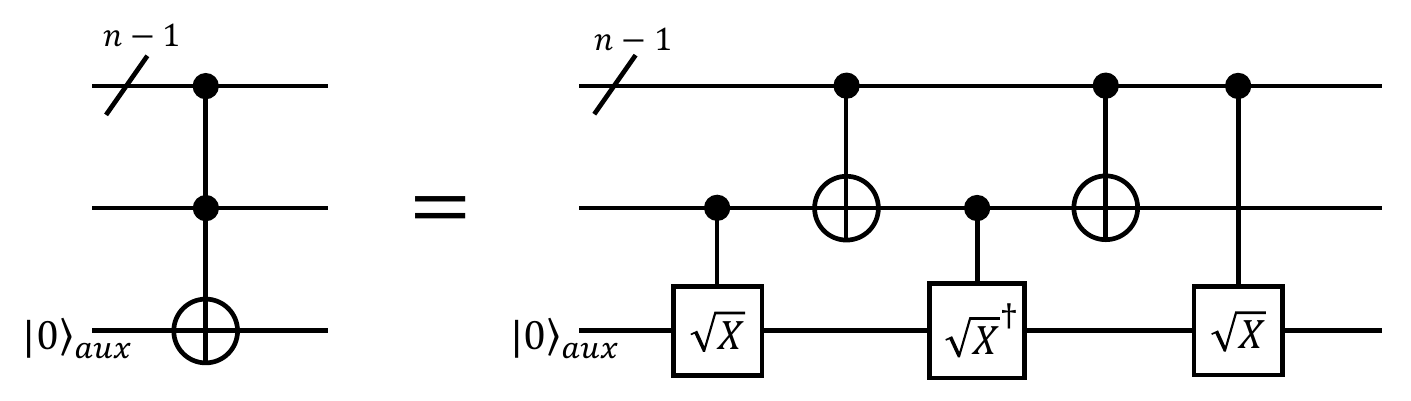}
    \caption{ Decomposition of an $n$-controlled $X$ gate ($C^{n}(X)$) using two $(n-1)$-controlled $X$ gates ($C^{n-1}(X)$), an $(n-1)$-controlled $\sqrt{X}$ gate ($C^{n-1}(\sqrt{X})$), and two single-controlled gates. Here, $\sqrt{X}$ is the square root of $X$, and $\sqrt{X}^\dagger $ is the Hermitian conjugate of $\sqrt{X} $. Adapted from Lemma 7.5 of Ref.~\cite{barenco1995elementary}.}
    \label{fig4}
\end{figure}

\begin{figure}
    \centering
    \includegraphics[width=0.7\linewidth]{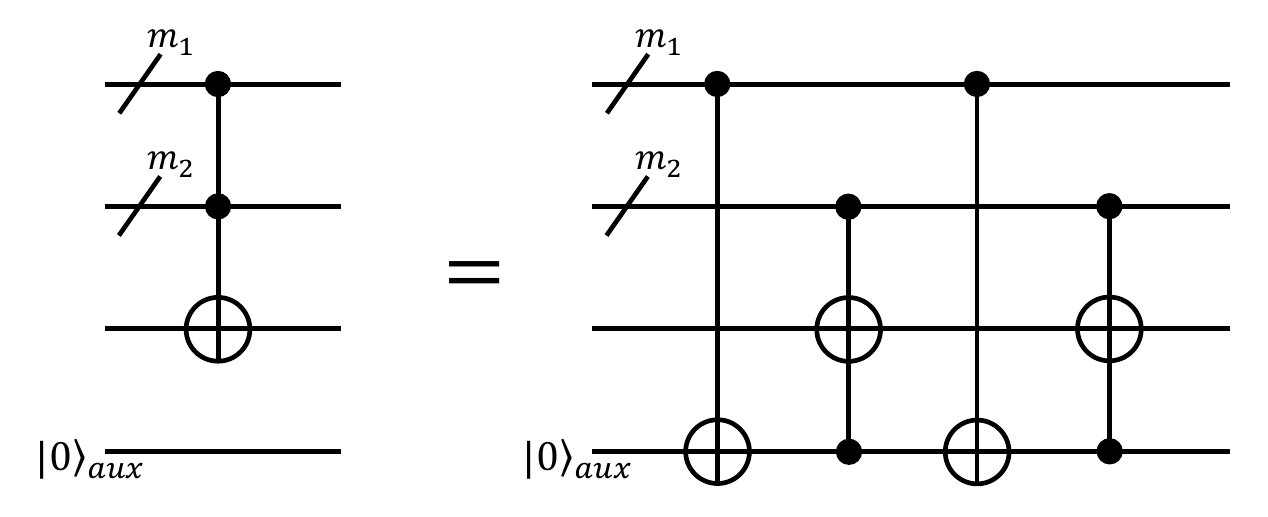}
    \caption{ Decomposition of a $C^{n-1}(X)$ gate using two $m_1$-controlled $X$ gates ($C^{m_1}(X)$) and two $(m_2+1)$-controlled $X$ gates ($C^{m_2+1}(X)$), where $m_1 = \left \lceil \frac{n}{2}  \right \rceil $ and $m_2 = n- m_1-1$. Adapted from Lemma 7.3 of Ref.~\cite{barenco1995elementary}.}
    \label{fig5}
\end{figure}

The next step is to decompose the $C^{n-1}(X)$ gates and $C^{n-1}(\sqrt{X})$ gates produced in the previous step. For the $C^{n-1}(X)$ gate, a naive strategy is to apply the decomposition shown in Figure~\ref{fig4} iteratively. However, this approach results in an exponential growth of the circuit size (Lemma 7.1 of Ref.~\cite{barenco1995elementary}). To avoid this exponential overhead, one can instead employ a single auxiliary qubit to split the $C^{n-1}(X)$ gate into 4 half-sized MCX gates, which can then be further decomposed into Toffoli gates. The corresponding decomposition is illustrated in Figure~\ref{fig5}. Here, the auxiliary qubit is a borrowed auxiliary qubit (initially in an arbitrary state and restored to its original state after the operation), so that no additional auxiliary qubits are required.

Specifically, a $C^{n-1}(X)$ gate can be decomposed into two $m_1$-controlled $X$ gates ($C^{m_1}(X)$) and two $(m_2+1)$-controlled $X$ gates ($C^{m_2+1}(X)$), where $m_1 = \left \lceil \frac{n}{2}  \right \rceil $ and $m_2 = n- m_1-1$. Each $C^{m_1}(X)$ gate can then be further decomposed into $4\left \lceil \frac{n}{2}  \right \rceil-8$ Toffoli gates (Figure 7 in Ref.~\cite{liu2008analytic}), while each $C^{m_2+1}(X)$ gate can be decomposed into $4n-4\left \lceil \frac{n}{2}  \right \rceil-8$ Toffoli gates (Figure 8 in Ref.~\cite{liu2008analytic}). In these decompositions, qubits not involved in the current gate act as borrowed auxiliary qubits, and no additional auxiliary qubits are introduced. It is worth noting that $n \ge 6$ is the applicable condition for the above decompositions. As for the case where $n < 6$, Table 1 in Ref.~\cite{liu2008analytic} lists them.

As for the $C^{n-1}(\sqrt{X})$ gate, it can be further decomposed by recursively applying the circuit identity shown in Figure~\ref{fig4}. This recursion is performed $n-2$ times until the $C^{n-1}(\sqrt{X})$ gate is decomposed into a set of Toffoli gates and single-controlled gates. The Toffoli gates can be implemented using the Congruent Modulo Phase Shifts (CMPS) method~\cite{barenco1995elementary}, which decomposes each Toffoli gate into 3 CNOT gates and 4 single-qubit gates, as illustrated in Figure~\ref{fig6}. The single-controlled gates can be decomposed into elementary gates according to the circuit identity shown in Figure~\ref{fig7}, where $\sqrt[n]{Z} $ denotes the $n$-th root of the Pauli-$Z$ operator, defined as $\sqrt[n]{Z} =\begin{bmatrix}
1  & 0\\
0  & e^{i\pi /n } 
\end{bmatrix}$.

\begin{figure}
    \centering
    \includegraphics[width=0.8\linewidth]{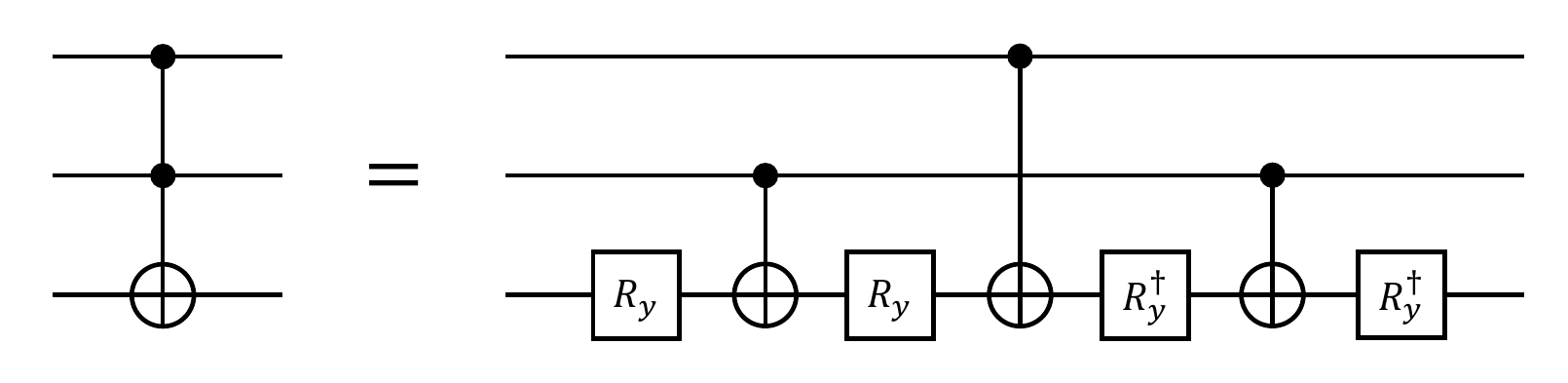}
    \caption{ Decomposition of a Toffoli gate using the Congruent Modulo Phase Shifts (CMPS) method,  where the angle of all $R_y$ ($R_y^{\dagger}$) operators is $\pi/4$. Adapted from Ref.~\cite{barenco1995elementary}.}
    \label{fig6}
\end{figure}

\begin{figure}
    \centering
    \includegraphics[width=0.95\linewidth]{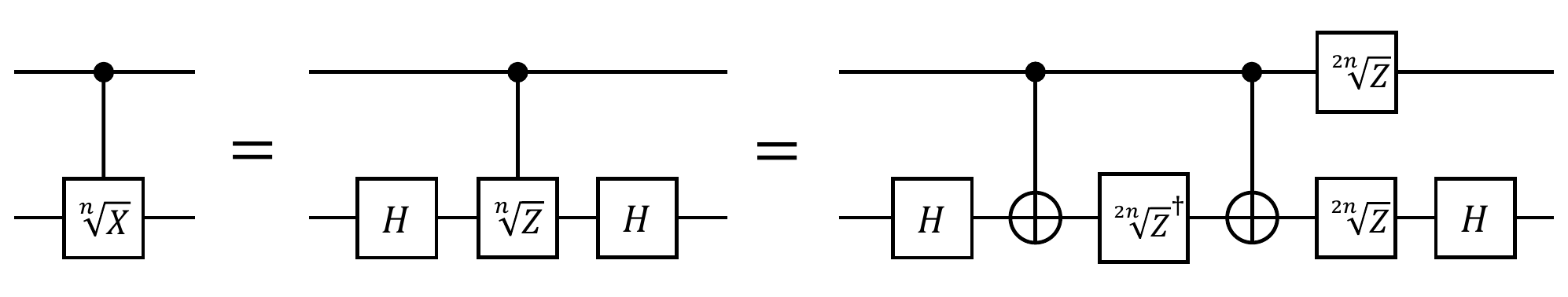}
    \caption{ Decomposition of a single-controlled gate using the circuit identity, where $\sqrt[n]{Z} $ denotes the $n$-th root of $Z$.}
    \label{fig7}
\end{figure}

Ref.~\cite{liu2008analytic} presents a cascaded structure for the iterative decomposition of an arbitrary multi-controlled unitary gate, which is equivalent to decomposing a $C^{n}(X)$ gate when the unitary matrix is $X$. With this decomposition method, when $n \ge 6$, a $C^{n}(X)$ gate can be decomposed into $24n^2-164n+352$ CNOT gates and $ 32n^2-224n+483$ single-qubit gates. It is important to note that the above decomposition assumes the absence of auxiliary qubits. If at least one auxiliary qubit is available, the circuit size of MCX gate decomposition can be reduced from $\mathcal{O}(n^2)$ to $\mathcal{O}(n)$~\cite{maslov2003improved,he2017decompositions}.

Combining the decomposition shown in Figure~\ref{fig3} with the exact number of elementary gates mentioned above, when $n \ge 6$, a multi-OR-controlled phase gate with $n$ control qubits can be decomposed into $48n^2-328n+704$ CNOT gates and $32n^2-222n+485$ single-qubit gates. For the cases $1<n<6$, the corresponding number of elementary gates can be calculated from Table 1 of Ref.~\cite{liu2008analytic}. We summarize these results in Table~\ref{tab2}.

\begin{table}[htbp]
\centering
\caption{ Number of elementary gates after decomposing different types of multi-controlled quantum gates}
\vspace{8pt}
\label{tab2}
\begin{tabular}{clcccccc}
\toprule
\multirow{1}{*}{\centering Gate Type} & \multirow{1}{*}{\centering Elementary Gate} & 2 & 3 & 4 & 5 & & \multicolumn{1}{c}{$n \geq 6$} \\
\midrule
\multirow{3}{*}{\makecell[l]{$n$-controlled $X$ gate (MCX)}} & CNOT  & 6 & 16 & 56 & 132 &  & $24n^2-164n+352$  \\
 & Single-qubit & 8 & 22 & 71 & 163 & & $32n^2-224n+483$    \\
 & Total & 14 & 38 & 127 & 295 & & $56n^2-388n+835$  \\
\\
\multirow{3}{*}{\makecell[l]{$n$-OR-controlled phase gate}} & CNOT  & 12 & 32 & 112 & 264 &  &$48n^2-328n+704$  \\
 & Single-qubit & 14 & 30 & 81 & 175 &  & $32n^2-222n+485$   \\
 & Total & 26 & 62 & 193 & 439 & & $80n^2-550n+1189$  \\
\bottomrule
\end{tabular}
\end{table}

\subsection{Gate complexity }
\label{sec4.2}

In this section, we analyze the gate complexity of the AQFG-QAOA algorithm for different graph structures. Specifically, we consider two representative graph classes: 3-regular graphs and \ER (ER) random graphs with an edge probability of 0.5. In a 3-regular graph, each vertex has a degree of exactly 3, forming a sparse structure when the number of vertices $n$ is large. In contrast, an ER graph $G(n,p_{e})$ exhibits a non-uniform degree distribution determined by the edge probability $p_{e}$~\cite{bollobas2011random}. For a fixed constant $p_{e}$, the expected degree of each vertex is $p_{e}(n-1)$, corresponding to an unstructured probabilistic connectivity model. For the MDS problem on an ER graph, the graph is extremely sparse when $p_{e}$ is close to 0, making the problem relatively easy to solve. When $p_{e}$ approaches 1, the graph approaches a complete graph, and the problem becomes simple again. The most challenging case occurs when $p_{e}$ takes a value between 0 and 1, at which the ER graph is neither too sparse nor too dense. From an information theory perspective, the Shannon entropy quantifies the uncertainty of the edge structure in a random graph, with higher entropy indicating greater structural complexity~\cite{anand2009entropy}. For ER graphs, the Shannon entropy of the edge set is maximized at $p_{e}=0.5$. Accordingly, in this work, we focus on ER graphs with edge probability $p_{e}=0.5$.

For a 3-regular graph with $n$ ($n \ge 4$) vertices, each vertex has 3 neighbors, plus the vertex itself, which corresponds to 4 control qubits in the multi-OR-controlled phase gate that implements the $T_k(x)$ clause. Consequently, each $T_k(x)$ clause is implemented by a 4-OR-controlled phase gate. According to Table~\ref{tab2}, such a gate can be decomposed into 112 CNOT gates and 81 single-qubit gates. Therefore, the circuit part that implements all $T_k(x)$ clauses requires a total of $112n$ CNOT gates and $81n$ single-qubit gates. In addition, using the decomposition shown in Figure~\ref{figA4}, all the single-controlled phase gates that implement all $D_k(x)$ clauses can be decomposed into $2n$ CNOT gates and $4n$ single-qubit gates. Including the $n$ Hadamard gates for preparing the initial state and the $n$ single-qubit $R_X(2\beta)$ gates in the mixing operator, the single-layer quantum circuit of the AQFG-QAOA algorithm on a 3-regular graph with $n$ vertices contains a total of $114n$ CNOT gates and $87n$ single-qubit gates.

For an ER graph $G(n, p_{e})$ with edge probability $p_{e} = 0.5 $, the average degree of each vertex is $\frac{n-1}{2}$, meaning each vertex is connected to an average of $\left \lceil \frac{n-1}{2}  \right \rceil $ other vertices. Including the vertex itself, a multi-OR-controlled phase gate that implements the $T_k(x)$ clause has $\left \lceil \frac{n+1}{2}  \right \rceil $ control qubits. When $n\ge 6$ and $n$ is odd, such a gate can be decomposed into $12n^2-140n+552$ CNOT gates and $8n^2-95n+382$ single-qubit gates. When $n\ge 6$ and $n$ is even, it can be decomposed into $12n^2-116n+424$ CNOT gates and $8n^2-79n+295$ single-qubit gates. Consequently, for $n\ge 6$, implementing all $n$ $T_k(x)$ clauses requires $12n^3-140n^2+552n$ CNOT gates and $8n^3-95n^2+382n$ single-qubit gates when $n$ is odd, and $12n^3-116n^2+424n$ CNOT gates and $8n^3-79n^2+295n$ single-qubit gates when $n$ is even. The number of elementary gates required to implement the $n$ $D_k(x)$ clauses, as well as those used for initial state preparation and the mixing operator, is identical to the corresponding counts for 3-regular graphs. Summing all gates yields the total number of elementary gates needed on an ER graph with edge probability $p_{e}=0.5$ using the AQFG-QAOA algorithm. For the cases $n < 6$, the gate counts must be calculated individually. A summary of these results is provided in Table~\ref{tab3}.

\begin{table}[htbp]
\centering
\caption{ Number of elementary gates after decomposition of a single-layer quantum circuit constructed using the AQFG-QAOA algorithm}
\vspace{8pt}
\label{tab3}
\begin{tabular}{ccccc}
\toprule

\multirow{2}{*}{\centering Graph Type} & \multirow{2}{*}{\centering Graph Size} & \multicolumn{3}{c}{Elementary Gate} \\
\cmidrule(lr){3-5} 
 &   & CNOT & Single-qubit & Total     \\
\midrule
\multirow{1}{*}{\makecell[c]{$n$-vertex 3-regular graph}}  & $n \ge 4$ & $114n$ & $87n$ & $201n$   \\
\\
\multirow{2}{*}{\makecell[c]{$n$-vertex ER graph \\ ($p_{e}=0.5$)}} & ($n \ge 6$, odd) & $12n^3-140n^2+554n$ & $8n^3-95n^2+388n$ & $20n^3-235n^2+942n$  \\
& ($n \ge 6$, even) & $12n^3-116n^2+426n$ & $8n^3-79n^2+301n$ & $20n^3-195n^2+727n$\\
\bottomrule
\end{tabular}
\end{table}

\subsection{Algorithm comparison}
\label{sec4.3}

Here, we compare the gate complexity of the AQFG-QAOA algorithm with that of the AQFH-QAOA algorithm proposed in Section~\ref{sec3}. The comparison is summarized in Table~\ref{tab4}. 

\begin{table}[htbp]
\centering
\caption{Number of elementary gates after decomposition of a single-layer quantum circuit constructed using different algorithms}
\vspace{8pt}
\label{tab4}
\begin{adjustbox}{width=\textwidth, center}
\begin{tabular}{clcccc}
\toprule
\multirow{2}{*}{\centering Graph Type} & \multirow{2}{*}{\centering \makecell[c]{Elementary \\Gate}} & \multicolumn{4}{c}{\centering Algorithm} \\
\cmidrule(lr){3-6} 
  &   &   \multicolumn{2}{c}{ \centering AQFG-QAOA} & \multicolumn{2}{c}{ \centering AQFH-QAOA }     \\
\midrule
& & \multicolumn{2}{c}{ \centering ($n \ge 4$) } &  \multicolumn{2}{c}{ \centering ($n \ge 4$) } \\
\multirow{3}{*}{\makecell[c]{$n$-vertex \\ 3-regular graph}} & CNOT & \multicolumn{2}{c}{\centering $114n$}  & \multicolumn{2}{c}{ \centering $34n$} \\
 & Single-qubit & \multicolumn{2}{c}{\centering $87n$} & \multicolumn{2}{c}{\centering $18n$}  \\
 & Total & \multicolumn{2}{c}{\centering $201n$} & \multicolumn{2}{c}{\centering $52n$}\\
\\
 & & ($n \ge 6$, odd) &  ($n \ge 6$, even) & ($n \ge 2$, odd) &  ($n \ge 2$, even) \\
\multirow{3}{*}{\makecell[c]{$n$-vertex ER \\ graph ($p_{e}=0.5$)}} & CNOT & $12n^3-140n^2+554n$ & $12n^3-116n^2+426n$  & $(n^2-3n) \cdot 2^{\frac{n-1}{2} }+2n$ & $(n^2-2n) \cdot 2^{\frac{n}{2}}+2n$\\
 & Single-qubit & $8n^3-95n^2+388n$ & $8n^3-79n^2+301n$  & $n\cdot 2^{\frac{n+1}{2}}-\frac{n^2}{2}+\frac{3n}{2}$ & $2n \cdot 2^{\frac{n}{2}}-\frac{n^2}{2}+n$ \\
 & Total & $20n^3-235n^2+942n$ & $20n^3-195n^2+727n$ & $(n^2-n) \cdot 2^{\frac{n-1}{2} }-\frac{n^2}{2}+\frac{7n}{2}$ &  $n^2 \cdot 2^{\frac{n}{2}}-\frac{n^2}{2}+3n$\\
\bottomrule
\end{tabular}
\end{adjustbox}
\end{table}

For 3-regular graphs, the AQFH-QAOA algorithm uses fewer CNOT gates and single-qubit gates than the AQFG-QAOA algorithm, which suggests that the MDS problem on 3-regular graphs is more suitable for solving using the AQFH-QAOA algorithm. For ER graphs with edge probability $p_{e}=0.5$, the number of elementary gates used in the AQFG-QAOA algorithm grows polynomially with the number of vertices, whereas the gate count of the AQFH-QAOA algorithm grows exponentially. Nevertheless, the practical applicability of each algorithm can be evaluated by estimating the maximum feasible problem size based on the number of CNOT gates.

\begin{figure}
    \centering
    \includegraphics[width=0.68\linewidth]{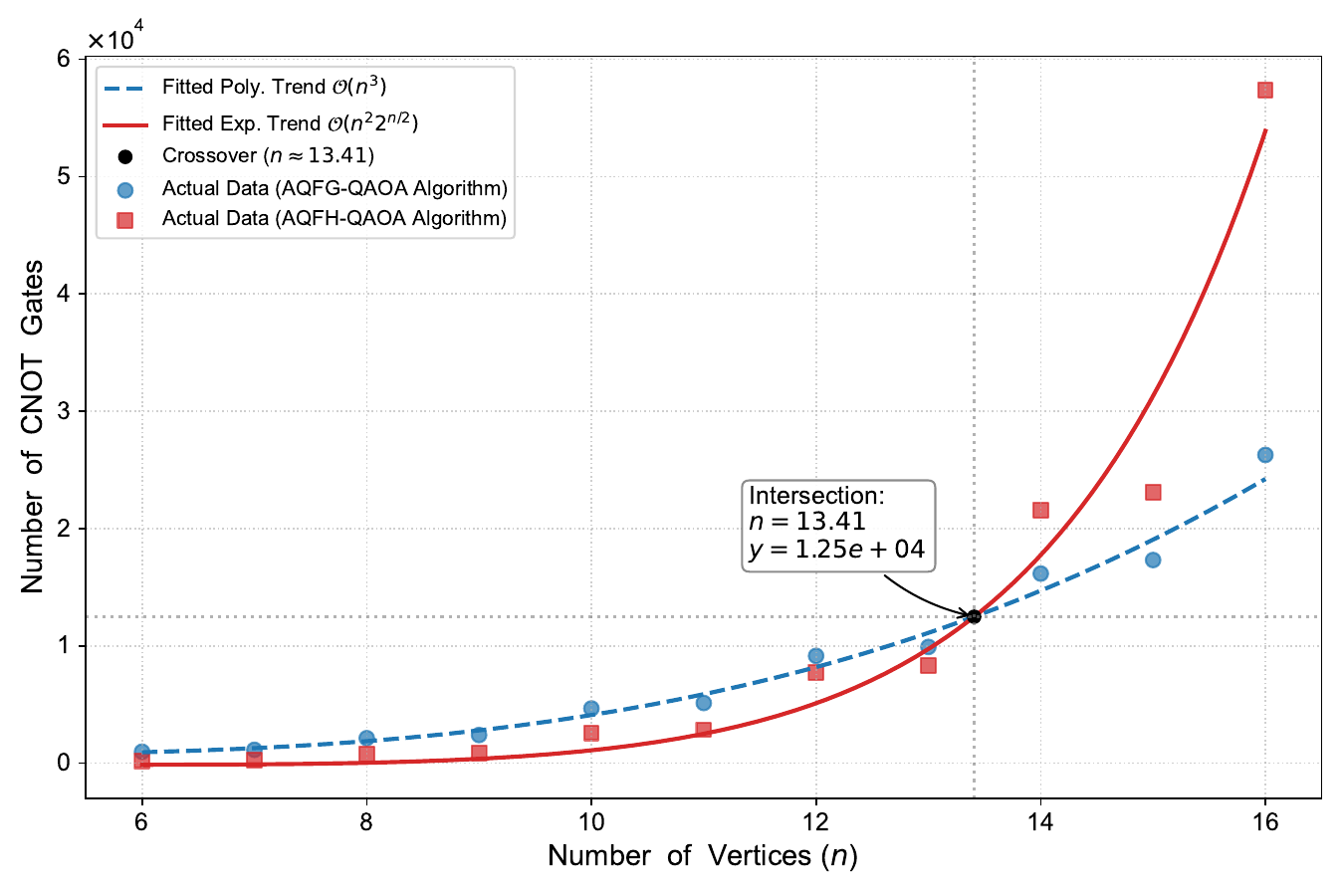}
    \caption{Number of CNOT gates in quantum circuits constructed using different algorithms to solve the MDS problem on ER graphs with different vertex sizes, where discrete data points represent actual data, and the corresponding curves show fitted growth trends based on these data.}
    \label{fig8}
\end{figure}

Figure~\ref{fig8} presents the number of CNOT gates required by the AQFG-QAOA and AQFH-QAOA algorithms for the MDS problem on ER graphs with edge probability $p_{e}=0.5$ at different vertex sizes. The discrete data points correspond to the exact values computed from Table~\ref{tab4}, while the solid curves represent fitted growth trends based on these data. As expected, the number of CNOT gates for both algorithms increases with the number of vertices $n$, but the growth rates differ substantially. Specifically, the number of CNOT gates in the AQFH-QAOA algorithm grows rapidly with $n$, exhibiting an exponential growth trend, whereas the AQFG-QAOA algorithm shows a much slower increase that is well approximated by a cubic polynomial. The fitted curves intersect at approximately $n\approx 13.41$, and the exact crossover point in terms of integer vertex size is found to satisfy $13 < n < 14$. This indicates that for ER graphs with vertex size $n \le 13$, the AQFH-QAOA algorithm requires fewer CNOT gates than the AQFG-QAOA algorithm, suggesting a potential advantage due to its lower circuit overhead. In contrast, for ER graphs with $n > 13$, the number of CNOT gates of the AQFG-QAOA algorithm is significantly lower than that of the AQFH-QAOA algorithm, demonstrating that the AQFG-QAOA algorithm has superior scalability.

As described in Section~\ref{sec4.2}, the average degree of each vertex in an $n$-vertex ER graph ($p_{e}=0.5$) is $d=\frac{n-1}{2}$. For such graphs, the AQFH-QAOA algorithm has low gate complexity when the vertex size $n \le 13$, and the average vertex degree in this case is $d=6$. Therefore, for ER graphs with $d < 6$, the AQFH-QAOA algorithm will have an advantage in gate complexity on graphs with larger vertex sizes.

\section{Numerical Simulations}
\label{sec5}

In this section, we use the MindSpore Quantum platform~\cite{xu2024mindspore} to perform numerical simulations. We compare the performance of the AQFH-QAOA algorithm with other related algorithms reviewed in Section~\ref{sec2.3}. The benchmark instances include 3-regular graphs and ER graphs with an edge probability of 0.5. Algorithmic performance is evaluated using the success probability defined in Definition~\ref{def1}.

\begin{definition}[Success Probability]
    The success probability is defined as the probability that the final state $|\psi_p (\vec{\vgamma}, \vec{\vbeta}) \rangle$ of the quantum circuit collapses to the ground state $|x^* \rangle$ of the objective Hamiltonian upon a single measurement, i.e., 
    \begin{equation}
    P_{\mathrm{suc}} = |\langle x^* |\psi_p ( \vec{\vgamma} ,\vec{\vbeta} ) \rangle | ^{2} .
    \end{equation}
    \label{def1}
\end{definition}

The success probability quantifies the overlap between the final state $|\psi_p(\vec{\vgamma},\vec{\vbeta} )\rangle$ and the ground state $\ket{x^*}$. A larger value of $P_{\mathrm{suc}}$ indicates better algorithmic performance and a higher likelihood of obtaining a high-quality solution.

\subsection{Sensitivity analysis of the penalty coefficient}
\label{sec5.1}

In the AQFH-QAOA algorithm, the choice of the penalty coefficient $\lambda$ is crucial for balancing constraint enforcement and objective optimization. Generally, if $\lambda$ is too small, the penalty term may fail to suppress infeasible solutions that violate the dominance constraints. Conversely, excessively large values of $\lambda$ may distort the optimization landscape and make the classical optimization process more difficult. Therefore, it is important to investigate the influence of the penalty coefficient $\lambda$ on the algorithm performance.

To this end, we perform numerical experiments by varying the penalty coefficient over the set $\lambda \in \{1.0, 1.2, 1.5, \\ 2.0, 3.0, 5.0\}$, where $\lambda = 1.0$ is used as the baseline. Two types of graph instances are considered: a 6-vertex 3-regular graph and a set of 20 \ER graphs with six vertices. For the ER graphs, we report the average success probability, which is defined as the mean of the success probabilities over all graph instances with the same number of vertices.

\begin{figure}
    \centering
    \includegraphics[width=\linewidth]{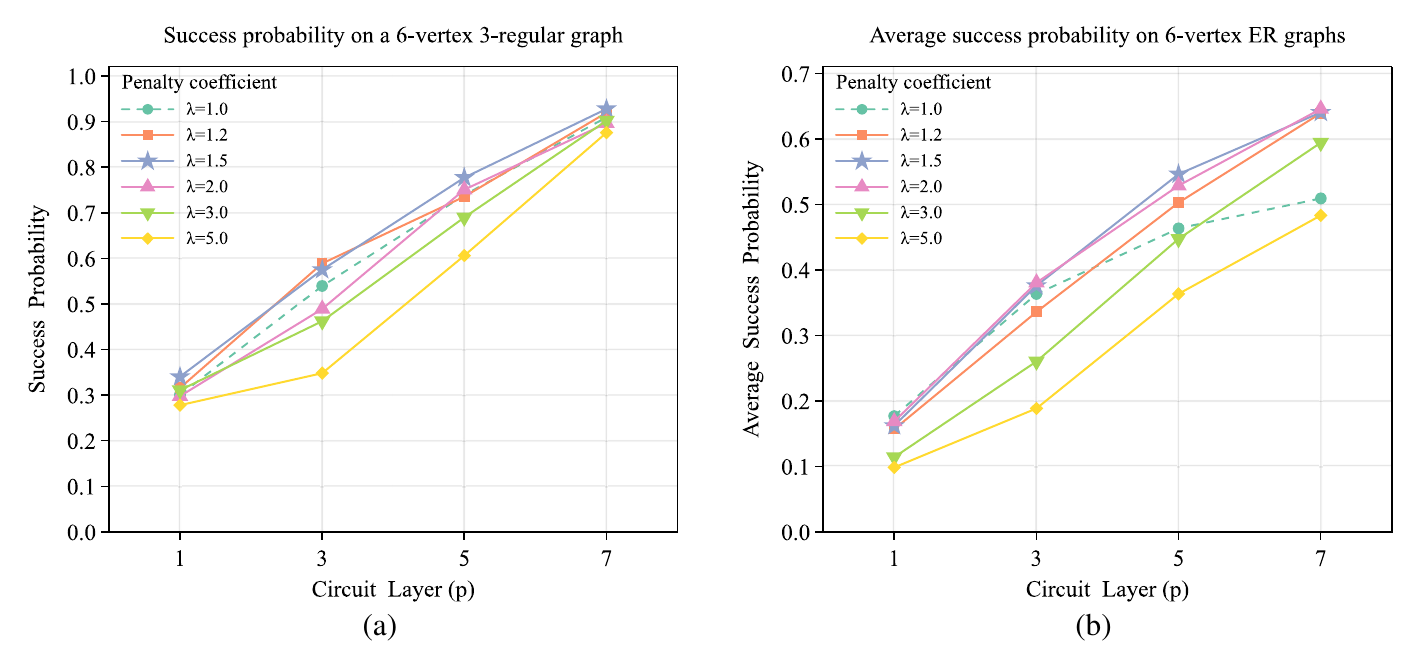}
    \caption{Success probabilities for different penalty coefficients on different types of graphs. Subgraph (a) shows the success probability on a 6-vertex 3-regular graph, while subgraph (b) shows the average success probability on 20 ER graphs.}
    \label{fig9}
\end{figure}

Figure~\ref{fig9}(a) shows the success probability of the AQFH-QAOA algorithm on a 6-vertex 3-regular graph for different values of $\lambda$. For all tested penalty coefficients, the success probability increases steadily with the circuit layer, indicating that deeper circuits generally improve the solution quality. However, noticeable differences appear among different values of $\lambda$. In particular, when $\lambda = 5.0$, the success probability is consistently lower than that obtained with the other choices for small and moderate circuit depths, which suggests that an excessively large penalty coefficient may reduce the effectiveness of the optimization. In contrast, moderate values such as $\lambda = 1.5$ and $\lambda = 2.0$ achieve higher success probabilities across most circuit layers. When the circuit layer becomes sufficiently large (e.g., $p=7$), the performance differences among most penalty coefficients become smaller, and the success probabilities converge to similar values.

Figure~\ref{fig9}(b) presents the average success probability over 20 ER graphs. A similar trend can be observed. The success probability increases monotonically with the circuit layer for all tested values of $\lambda$. For smaller penalty coefficients ($\lambda=1.2$), the performance is slightly worse than that of moderate values such as $\lambda=1.5$ and $\lambda=2.0$. Meanwhile, when $\lambda$ becomes too large (e.g., $\lambda=5.0$), the algorithm again exhibits reduced performance across different circuit layers. Overall, the results indicate that moderate penalty coefficients provide better and more stable performance across different graph instances.

These observations suggest that the AQFH-QAOA algorithm is relatively robust with respect to the choice of $\lambda$ within a reasonable range. In particular, moderate values of $\lambda$ (e.g., around $\lambda \approx 1.5$--$2.0$) tend to achieve the best performance for the tested graph instances, while excessively small or large penalty coefficients may lead to slightly degraded results. In the subsequent simulations, we adopt a fixed penalty coefficient $\lambda = 1.5$.

\subsection{Performance comparison of related algorithms}
\label{sec5.2}

We perform numerical simulations on randomly generated 3-regular graphs and ER graphs with an edge probability of 0.5, comparing the proposed algorithm with existing algorithms. The algorithms proposed by Dinneen~\textit{et al.} and Pan~\textit{et al.} require a large number of auxiliary qubits. For instance, both algorithms require 12 auxiliary qubits (18 qubits in total) to solve the MDS problem on a 6-vertex 3-regular graph, with even more auxiliary qubits needed for graphs with larger vertex sizes or higher edge densities. Consequently, due to the limitations of classical simulation, these two algorithms can only be evaluated on graphs with 4 and 6 vertices in our experiments. 

\begin{figure}[ht]
    \centering
    \includegraphics[width=0.92\linewidth]{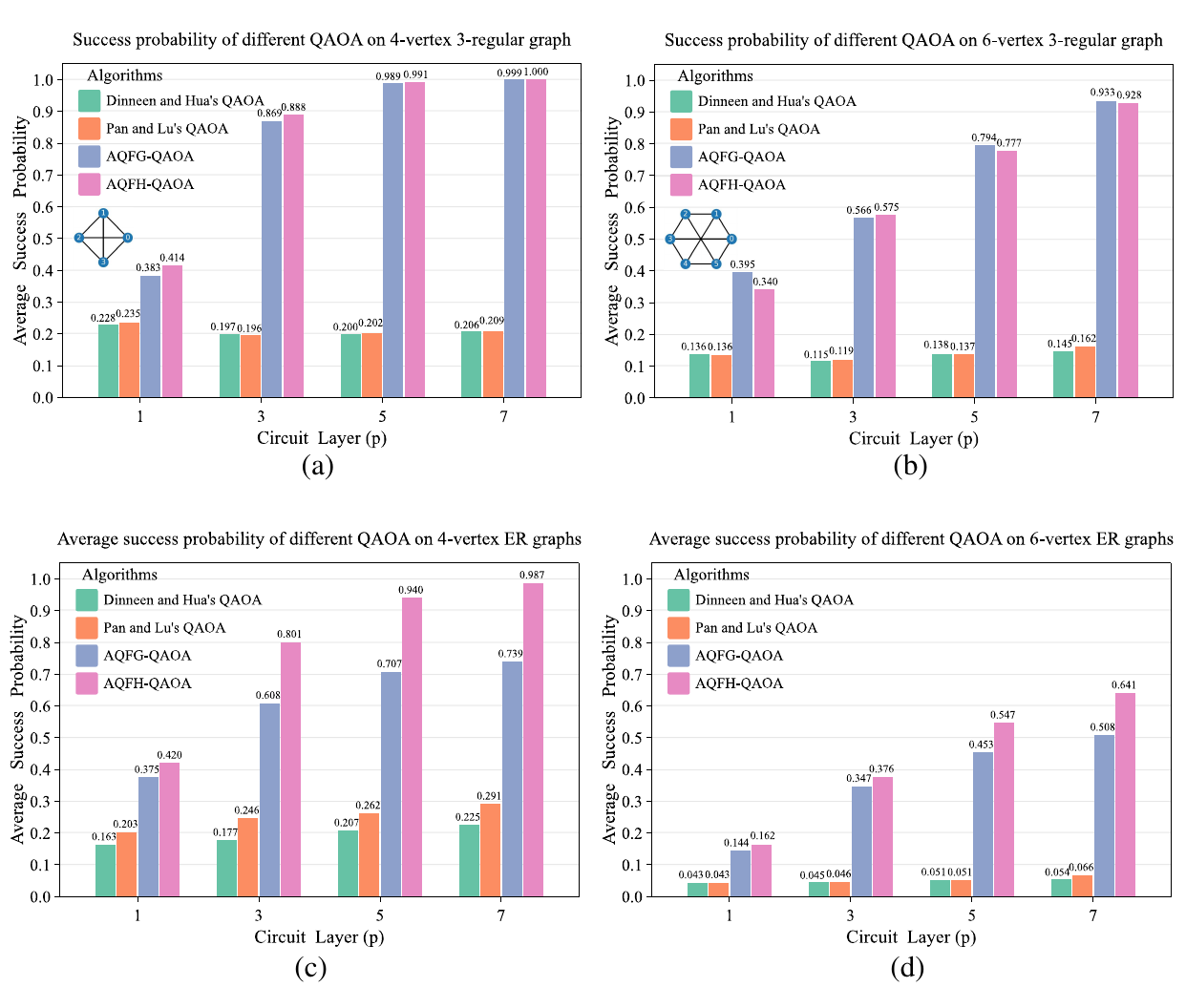}
    \caption{Success probabilities of related algorithms for the MDS problem on 3-regular and ER graphs. Here, subfigures (a) and (b) present the success probabilities on 3-regular graphs with 4 and 6 vertices, respectively. Subfigures (c) and (d) show the average success probabilities on 10 randomly generated 4-vertex ER graphs and 20 randomly generated 6-vertex ER graphs, respectively. }
    \label{fig10}
\end{figure}

Figure~\ref{fig10} demonstrates that the AQFH-QAOA algorithm consistently achieves higher success probabilities than the algorithms of Dinneen \textit{et al.} and Pan \textit{et al.} across all tested graph instances. Moreover, as the number of circuit layers $p$ increases, the success probability of the AQFH-QAOA algorithm improves significantly, while the success probabilities of the other two algorithms either fluctuate or increase only slightly, leading to an increasingly pronounced performance gap. 

Compared with the AQFG-QAOA algorithm, the AQFH-QAOA algorithm generally exhibits superior performance, with the advantage becoming more pronounced as $p$ increases. However, for the 6-vertex 3-regular graph, it performs slightly worse than the AQFG-QAOA algorithm. This behavior can be attributed to circuit depth. Although the quantum circuits of the AQFG-QAOA algorithm contain more elementary gates on small graphs, classical numerical simulations do not decompose the multi-OR-controlled phase gates in the circuits due to the limitations of the simulation platform. In contrast, the quantum circuits of the AQFH-QAOA algorithm are implemented using elementary gates through the unitary evolution of Hamiltonians. This resulted in a significantly smaller actual circuit depth of the AQFG-QAOA algorithm compared to the AQFH-QAOA algorithm for the same problem instances, which can lead to improved algorithmic performance in simulation.

Overall, these results indicate that the AQFH-QAOA algorithm significantly outperforms the QAOA algorithms of Dinneen~\textit{et al.} and Pan~\textit{et al.}. However, due to the absence of multi-OR-controlled phase gate decomposition in the simulations, its performance is comparable to or slightly worse than that of the AQFG-QAOA algorithm in some cases.

\subsection{Algorithm comparison and ablation study}
\label{sec5.3}

We further compared the AQFH-QAOA algorithm with the AQFG-QAOA algorithm on 3-regular graphs and ER graphs with 8, 10, and 12 vertices. As discussed in Section~\ref{sec5.2}, the circuit depth of the AQFG-QAOA algorithm remains much smaller for these instances because its multi-OR-controlled phase gates are not decomposed in simulation. To explore whether algorithmic performance could be improved by enhancing the ansatz, we design an ablation study. Inspired by the Multi-Angle QAOA (MA-QAOA)~\cite{herrman2022multi}, we introduce multi-angle parameters into the quantum circuit of the AQFH-QAOA algorithm, creating a variant denoted as the AQFH-MA-QAOA algorithm. This modification is not readily applicable to AQFG-QAOA due to its fixed ansatz design~\cite{guerrero2020solving}. The algorithms of Dinneen \textit{et al.} and Pan \textit{et al.} are excluded from this comparison due to their previously observed poor performance and excessive auxiliary qubit requirements, which exceeded our simulation capabilities.

A brief overview of MA-QAOA is provided in Appendix~\ref{secB}. In the ablation study, we retained the quantum circuit structure of the AQFH-QAOA algorithm and only modified its parameters. Unlike the standard QAOA described in Section~\ref{sec2.2}, we assign an independent angle parameter to each parameterized quantum gate corresponding to the phase separation operator $U_P(\gamma)$ and the mixing operator $U_M(\beta)$. 

\begin{figure}[ht]
    \centering
    \includegraphics[width=\linewidth]{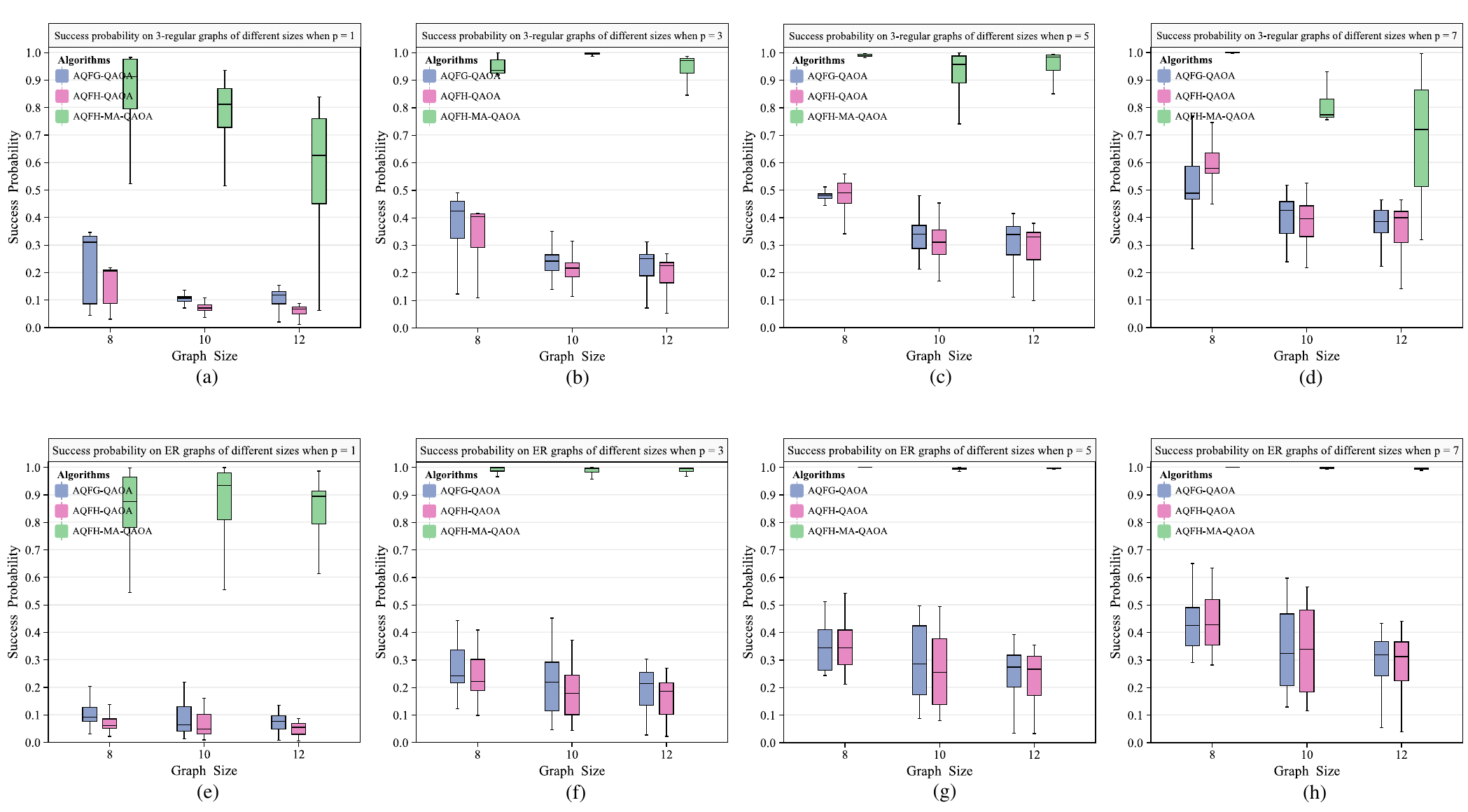}
    \caption{Success probabilities of different algorithms for the MDS problems on 3-regular and ER graphs with different numbers of circuit layers. Here, subfigures (a)-(d) correspond to the success probabilities on 3-regular graphs at circuit layers $p=1, 3, 5, 7$, respectively. And subfigures (e)-(h) correspond to the success probabilities on ER graphs at circuit layers $p=1, 3, 5, 7$, respectively. There are 20 graphs for each vertex size, and outlier data are omitted for clarity.}
    \label{fig11}
\end{figure}

Figure~\ref{fig11} displays the distribution of success probabilities across 20 graph instances. Several key trends can be observed from the results. First, the AQFH-MA-QAOA algorithm consistently achieves higher success probabilities across all tested circuit layers and graph types, outperforming both the AQFH-QAOA and AQFG-QAOA algorithms. Second, the success probabilities of AQFH-QAOA and AQFG-QAOA decline significantly as the graph size increases, reflecting the growing complexity of the problem. In contrast, the AQFH-MA-QAOA algorithm exhibits stable performance on most of the graph instances tested. Third, the advantage of the AQFH-MA-QAOA algorithm is particularly pronounced on ER graphs. For the circuit layer $p \geq 3$, it achieves both higher median success probabilities and tighter interquartile ranges, suggesting stronger robustness to variations in graph structure. 

These results demonstrate that the AQFH-MA-QAOA algorithm enhances the expressivity and adaptability of the quantum circuit, enabling a more effective exploration of the solution space. Overall, the AQFH-MA-QAOA algorithm provides higher solution quality than existing algorithms for the MDS problem and remains competitive as the problem size grows, making it a promising candidate for implementation on near-term quantum devices.

\section{Conclusions}
\label{sec6}

In this study, we provide two auxiliary-qubit-free QAOA algorithms for solving the MDS problem, namely AQFH-QAOA and AQFG-QAOA. The AQFH-QAOA algorithm reformulates the inequality constraints of the MDS problem into equality constraints using Boolean algebra, avoiding the use of auxiliary qubits in the circuit construction. Specifically, the inequality constraints are first reformulated as equivalent Boolean expressions and then converted into arithmetic expressions using generalized Boolean algebra identities. These arithmetic expressions are subsequently incorporated into the objective function as penalty terms to construct the objective Hamiltonian, which allows quantum circuits to be constructed based on the unitary evolution of the Hamiltonian without the need for auxiliary qubits. In contrast, existing QAOA algorithms typically rely on a large number of surplus variables to handle inequality constraints, which inevitably increases circuit overhead by introducing auxiliary qubits. Theoretical analysis demonstrates that this algorithm is particularly suitable for sparse graphs with small vertex degrees, not only reducing qubit overhead but also having lower gate complexity.

Furthermore, we utilize the decomposition technique of multi-controlled $X$ gates to provide an auxiliary-qubit-free optimized implementation method for Guerrero's QAOA, which is called the AQFG-QAOA algorithm. A detailed gate complexity analysis of the two algorithms shows that the AQFH-QAOA algorithm is well-suited for MDS problems on arbitrary 3-regular graphs, while for ER graphs with an edge probability of 0.5, the preferred algorithm depends on the graph size. Numerical simulations on randomly generated 3-regular and ER graphs demonstrate that the AQFH-QAOA algorithm achieves higher solution quality than most existing QAOA algorithms. However, due to the inability of the simulation platform to decompose multi-OR controlled phase gates in AQFG-QAOA, its performance is comparable to or slightly better than the AQFH-QAOA algorithm. Nevertheless, an enhanced variant of AQFH-QAOA based on the multi-angle QAOA framework, referred to as AQFH-MA-QAOA, is shown to effectively exploit the superior expressive power of independent parameters at low circuit depths, providing a competitive solution to the MDS problem on near-term quantum devices.

Although the proposed algorithms are evaluated using ideal simulations, their performance on real quantum devices may be affected by quantum noise. In particular, the circuit depth increases with the vertex degree, which may lead to the accumulation of gate errors, especially for two-qubit gates such as CNOT gates. Nevertheless, the proposed auxiliary-qubit-free algorithm significantly reduces the total number of qubits compared with previous approaches. Since qubit decoherence and readout errors typically scale with the number of physical qubits, this reduction in qubit count may partially mitigate the impact of increased circuit depth. Therefore, the proposed algorithm has the potential to provide a trade-off between circuit width and circuit depth for implementation on near-term quantum devices.

Looking ahead, this study can be expanded in two directions. First, it would be valuable to implement and benchmark the two algorithms on real quantum hardware to assess their practical performance under realistic noise conditions. Second, the proposed Boolean algebra-based approach for handling inequality constraints can be generalized to other combinatorial optimization problems with inequality constraints, thereby enabling the design of more resource-efficient quantum algorithms.

\section*{Data availability statement}
Data available on request from the authors.

\addcontentsline{toc}{chapter}{Acknowledgment}
\section*{Acknowledgment}
We thank Jiacheng Fan, Lingxiao Li, and Jing Li for their useful discussions on the subject. This work is supported by the National Natural Science Foundation of China (Grant Nos. 62372048, 62272056, 62371069, 62171056) and the National Key Laboratory of Security Communication Foundation (2025, 6142103042503).

\begin{appendices}
\appendix 

\renewcommand{\thefigure}{\thesection\arabic{figure}}\setcounter{figure}{0}

\section{Guerrero's gate decomposition method}
\label{secA}

Guerrero provides an explicit decomposition of multi-OR-controlled phase gates in Ref.~\cite{guerrero2020solving}. Specifically, such a gate can be decomposed into two multi-OR-controlled $X$ gates and a single-controlled phase gate by using an auxiliary qubit, as illustrated in Figure~\ref{figA1}. Subsequently, each multi-OR-controlled $X$ gate with $n$ control qubits can be further decomposed into $2n-3$ OR-controlled Toffoli gates using $n-2$ auxiliary qubits, as shown in Figure~\ref{figA2}. Each OR-controlled Toffoli gate can be implemented using one standard Toffoli gate and two CNOT gates, as depicted in Figure~\ref{figA3}. Furthermore, Figure~\ref{figA4} illustrates a method for decomposing a single-controlled phase gate with $|0\rangle$ as the control qubit into elementary gates consisting of CNOT gates and single-qubit gates. As shown, such a gate can be implemented using two CNOT gates and four single-qubit gates.

\begin{figure}[h]
    \centering
    \includegraphics[width=0.65\linewidth]{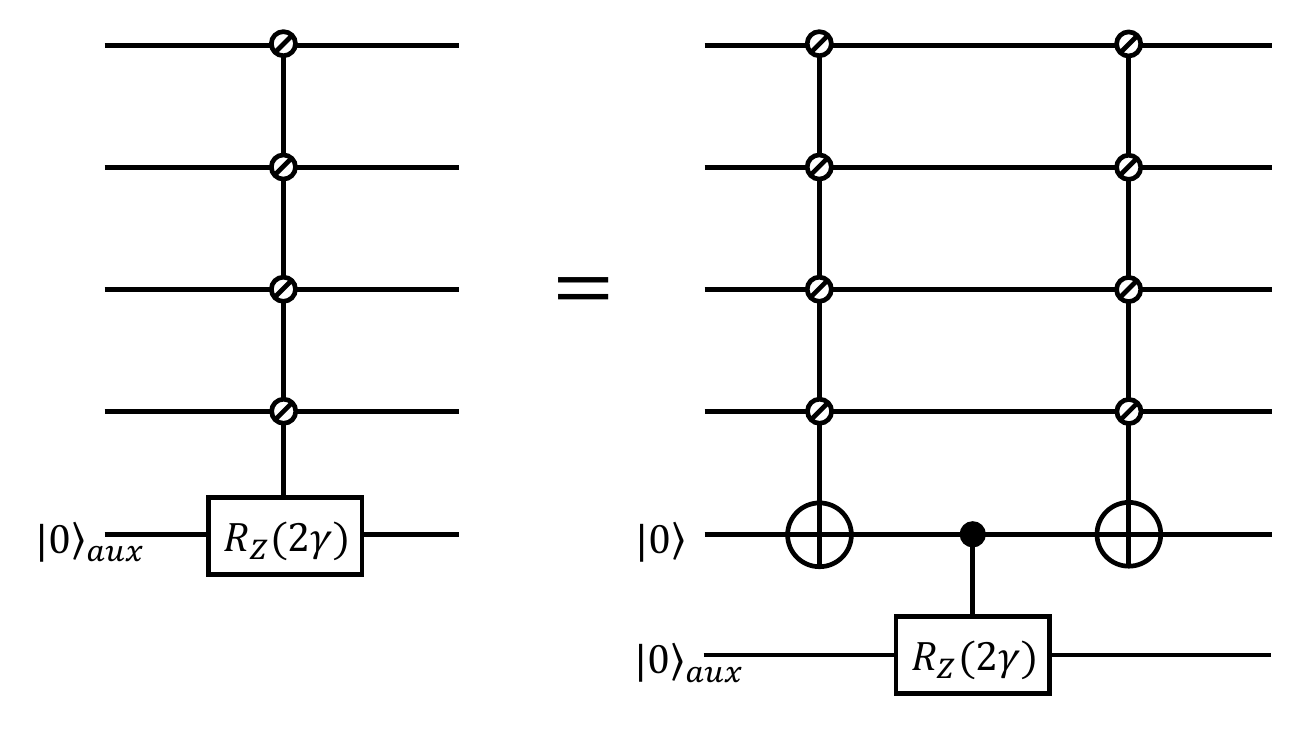}
    \caption{ Decomposition of a multi-OR-controlled phase gate with 4 control qubits using an auxiliary qubit and two multi-OR-controlled $X$ gates. Adapted from Figure 15 of Ref.~\cite{guerrero2020solving}.}
    \label{figA1}
\end{figure}

\begin{figure}[h]
    \centering
    \includegraphics[width=0.65\linewidth]{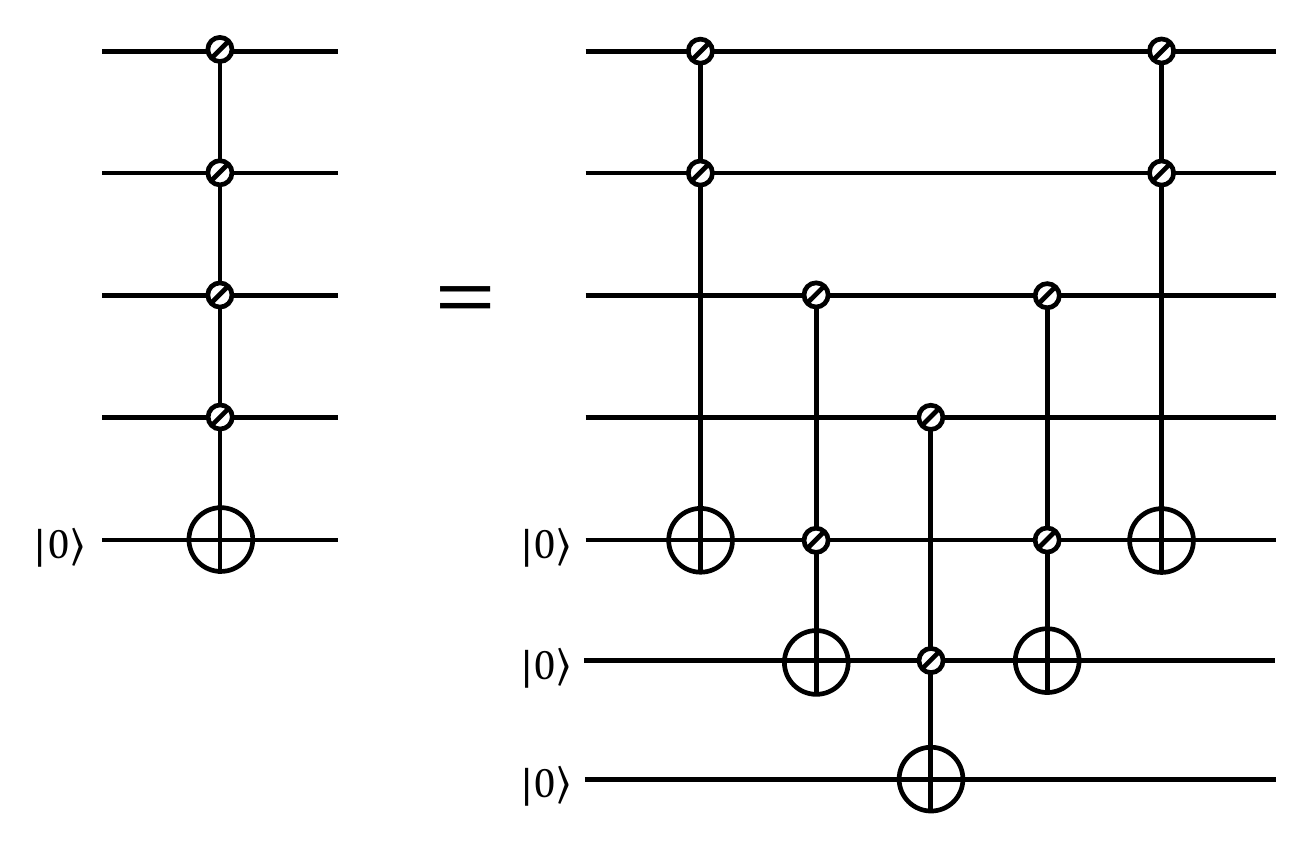}
    \caption{ Decomposition of a multi-OR-controlled $X$ gate with $n$ control qubits using $n-2$ auxiliary qubits and $2n-3$ OR-controlled Toffoli gates. Here, $n=4$, adapted from Figure 13 of Ref.~\cite{guerrero2020solving}.}
    \label{figA2}
\end{figure}

\begin{figure}
    \centering
    \includegraphics[width=0.5\linewidth]{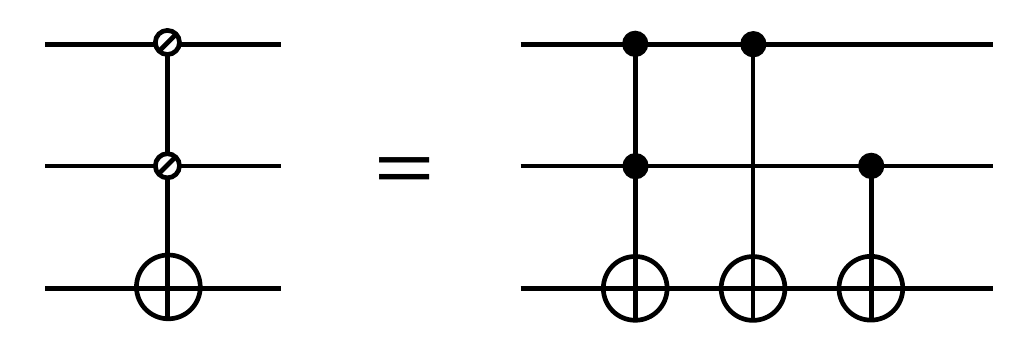}
    \caption{ Decomposition of an OR-controlled Toffoli gate using a single Toffoli gate and two CNOT gates. Adapted from Figure 11 of Ref.~\cite{guerrero2020solving}.}
    \label{figA3}
\end{figure}

\begin{figure}
    \centering
    \includegraphics[width=0.9\linewidth]{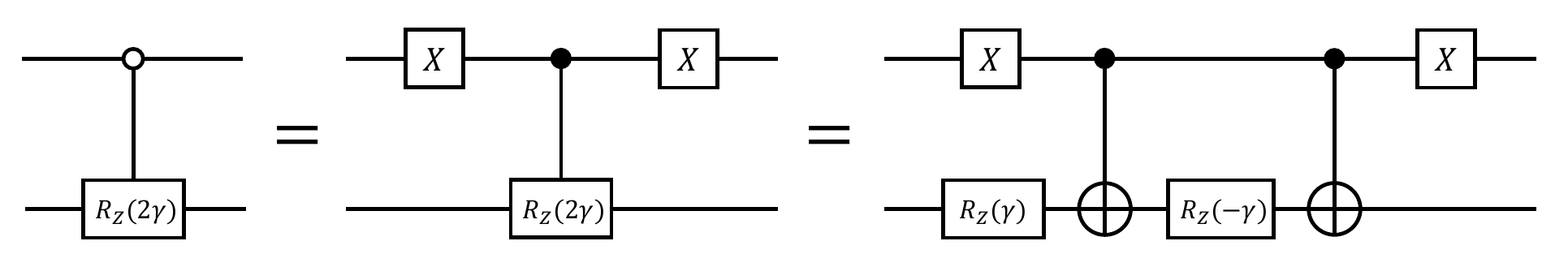}
    \caption{ Decomposition of a single-controlled phase gate using CNOT gates and single-qubit gates. }
    \label{figA4}
\end{figure}

According to Guerrero's decomposition method, an $n$-OR-controlled phase gate can be decomposed into $4n-6$ Toffoli gates, $8n-12$ CNOT gates, and a single-controlled phase gate, using a total of $n-1$ auxiliary qubits. These gates can be further decomposed into elementary gates, facilitating their implementation on quantum hardware. A complementary discussion of this decomposition scheme for multi-OR-controlled phase gates can also be found in Appendix E of Ref.~\cite{donkers2022qpack}. 

Next, we explicitly determine the number of elementary gates and auxiliary qubits required to decompose an $n$-vertex undirected graph using the methods described above. For simplicity, we assume that the average degree of the vertices in the undirected graph is $d$. For such a graph, each vertex has $d$ neighbors, and the multiple-OR-controlled phase gate that implements the corresponding $T_k(x)$ clause contains $d+1$ control qubits (including the vertex itself). As described above, such a $(d+1)$-OR-controlled phase gate can be decomposed into $4d-2$ Toffoli gates, $8d-4$ CNOT gates, and a single-controlled phase gate, using $d$ auxiliary qubits. A Toffoli gate can be further decomposed into 3 CNOT gates and 4 single-qubit gates using the congruent modulo phase shifts method~\cite{barenco1995elementary}, while a single-controlled phase gate can be decomposed into 2 CNOT gates and 2 single-qubit gates, as illustrated in Figure~\ref{figA4}. Consequently, a $(d+1)$-OR-controlled phase gate can be implemented using $20d-8$ CNOT gates and $16d-6$ single-qubit gates, together with $d$ auxiliary qubits. Since the degree of each vertex is assumed to be $d$, there are $n$ such $(d+1)$-OR-controlled phase gates in total. Therefore, implementing all $T_k(x)$ clauses requires $20nd-8n$ CNOT gates, $16nd-6n$ single-qubit gates, and $d$ auxiliary qubits.

Using the decomposition method shown in Figure~\ref{figA4}, implementing all the single-controlled phase gates corresponding to $n$ $D_k(x)$ clauses requires $2n$ CNOT gates and $4n$ single-qubit gates. Including the $n$ Hadamard gates used to prepare the initial state and the $n$ single-qubit $R_X(2\beta)$ gates in the mixing operator, the single-layer quantum circuit of Guerrero's QAOA algorithm for an $n$-vertex undirected graph with average vertex degree $d$ contains a total of $20nd-6n$ CNOT gates and $16nd$ single-qubit gates, using $d$ auxiliary qubits.

\renewcommand{\theequation}{\thesection\arabic{equation}}\setcounter{equation}{0}

\renewcommand{\thefigure}{\thesection\arabic{figure}}\setcounter{figure}{0}

\section{Introduction to Multi-angle QAOA}
\label{secB}

The Multi-Angle Quantum Approximate Optimization Algorithm (MA-QAOA) is a generalization of standard QAOA that introduces a separate angle parameter for each term in the phase separation operator and the mixing operator, instead of using a single parameter per operator per layer~\cite{herrman2022multi}. This enhanced parametrization improves the expressive power of the quantum circuit, enabling it to explore a wider range of Hilbert space and potentially yielding higher-quality solutions.

In the standard QAOA with a circuit layer $p$, the variational quantum state is generated by alternately applying the unitary evolution operators of the objective Hamiltonian $H_{P}$ and the mixing Hamiltonian $H_{M}$, as follows:
\begin{equation}
    \ket{\psi_p(\vec{\vgamma},\vec{\vbeta})} = \prod_{k=1}^p e^{-i\beta_k H_M} e^{-i \gamma_k H_P} \ket{\psi_0},
    \label{app eq: standard ansatz state}
\end{equation}
where $\vec{\vgamma} = (\gamma_1,\ldots,\gamma_p)$ and $\vec{\vbeta} = (\beta_1,\ldots,\beta_p)$ are the variational parameters. The $p$-layer circuit of the standard QAOA contains $2p$ parameters $\left \{ \gamma_{k}, \beta_{k} \right \}_{k=1}^{p} $. Studies have shown that fewer parameters may limit the ability of quantum circuits to explore the solution space, potentially leading to suboptimal solutions~\cite{herrman2022multi,vijendran2024expressive}.

MA-QAOA overcomes this limitation by assigning independent parameters to each parameterized quantum gate in
the circuit. Specifically, since the objective Hamiltonian $H_{P}$ and the mixing Hamiltonian $H_{M}$ are both sums of matrices, Herrman \textit{et al.} express them as $H_P = \sum_{u} C_{u}$ and $H_M = \sum_{v} B_v$. The phase separation operator of the $k$-th layer is reformulated as 
\begin{equation}
    U_{P}\left(\vec{\vgamma}_{k}\right) = e^{-i\vec{\vgamma}_{k} \sum_{u} C_{u}}= e^{-i \sum_{u} \gamma_{k, u} C_{u}} = \prod_{u} e^{-i \gamma_{k, u} C_{u}} ,
    \label{app eq: ma-qaoa phase separation operator}
\end{equation}
where $U_{P}\left(\vec{\vgamma}_{k}\right) = e^{-i\vec{\vgamma}_{k} H_{P}}$ is the unitary operator of the objective Hamiltonian $H_{P}$, with the parameter vector $\vec{\vgamma}_{k} = (\gamma_{k, u_{1}}, \gamma_{k, u_{2}}, \ldots)$. Similarly, the mixing operator becomes
\begin{equation}
    U_{M}\left(\vec{\vbeta}_{k}\right) = e^{-i\vec{\vbeta}_{k} \sum_{v} B_{v}}= e^{-i \sum_{v} \beta_{k, v} B_{v}} = \prod_{v} e^{-i \beta_{k, v} B_{v}}, 
    \label{app eq: ma-qaoa mixing operator}
\end{equation}
where $U_{M}\left(\vec{\vbeta}_{k}\right) = e^{-i\vec{\vbeta}_{k} H_{M}}$ is the unitary operator of the mixing Hamiltonian $H_{M}$, with the parameter vector $\vec{\vbeta}_{k} = (\beta_{k, v_{1}}, \beta_{k, v_{2}}, \ldots)$.

Hence, the ansatz state prepared by MA-QAOA with circuit layer $p$ is:
\begin{equation}
    \ket{\psi_p \Big(\{\vec{\vgamma}_{k}, \vec{\vbeta}_{k}\}_{k=1}^p \Big)} 
    = \prod_{k=1}^p \Big(\prod_{v} e^{-i \beta_{k,v} B_v} \prod_{u} e^{-i \gamma_{k,u} C_{u}} \Big) \ket{\psi_0}.
\end{equation}

By allowing each quantum gate to have an independent rotation angle, MA-QAOA generally achieves a higher success probability with fewer layers than standard QAOA~\cite{herrman2022multi}. Figure~\ref{figB1} illustrates the difference between the two algorithms, showing that the circuit structure is identical except for the parameters.

\begin{figure}[htbp]
    \centering
    \includegraphics[width=0.7\linewidth]{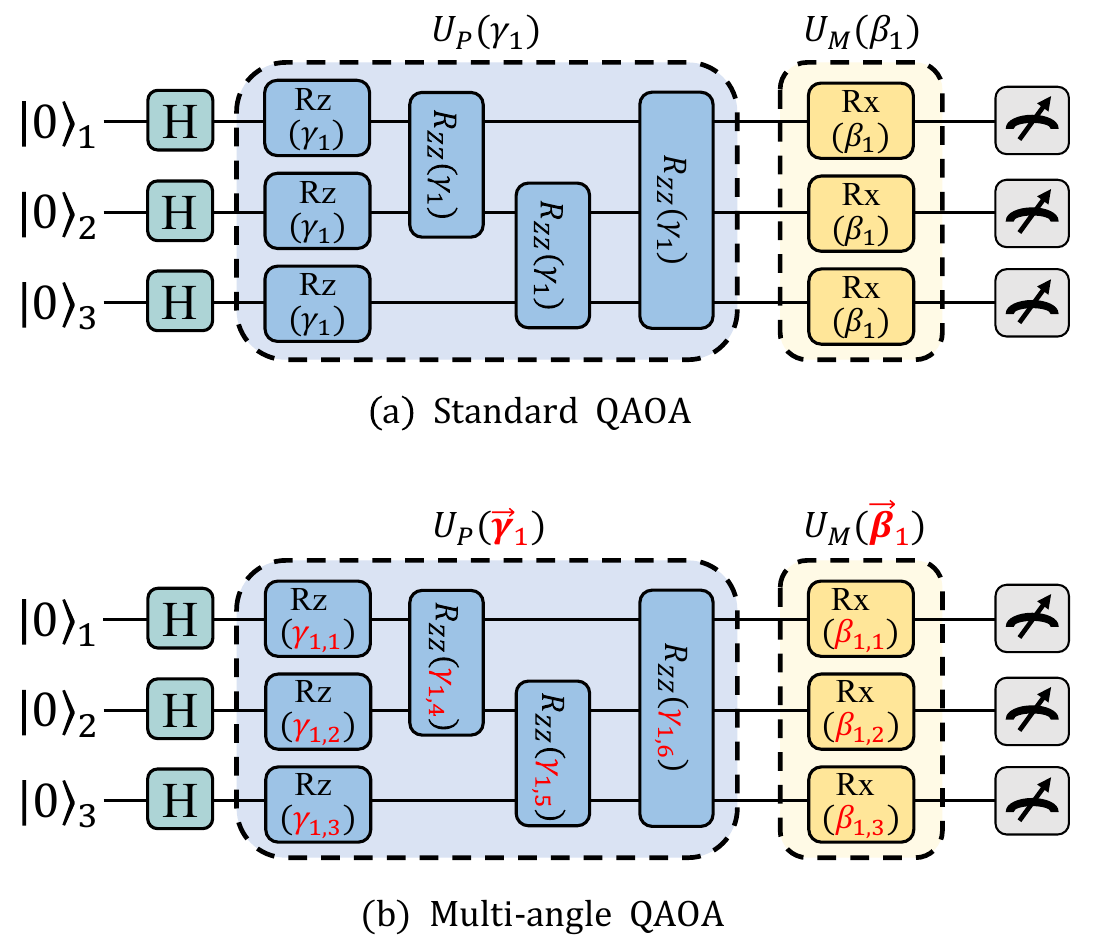}
    \caption{ Single-layer quantum circuits of (a) standard QAOA and (b) multi-angle QAOA. }
    \label{figB1}
\end{figure}

\end{appendices}

~\\
\addcontentsline{toc}{chapter}{References}


\providecommand{\newblock}{}

\end{document}